\documentclass[sigconf, nonacm]{acmart}

\newcommand\vldbdoi{XX.XX/XXX.XX}
\newcommand\vldbpages{XXX-XXX}
\newcommand\vldbvolume{20}
\newcommand\vldbissue{1}
\newcommand\vldbyear{2027}
\newcommand\vldbauthors{\authors}
\newcommand\vldbtitle{\shorttitle} 
\newcommand\vldbavailabilityurl{https://github.com/meimeimeiarc/PP-RFANNS}
\newcommand\vldbpagestyle{plain} 
\newcommand{\eat}[1]{}

\usepackage{bm}
\usepackage{algorithmic}
\usepackage[linesnumbered,ruled,vlined]{algorithm2e}
\usepackage{pifont}
\usepackage{enumitem}
\usepackage{footmisc}

\usepackage{graphicx}
\usepackage{booktabs}
\usepackage{makecell}

\begin{document}
\title{Efficient Privacy-Preserving Range Filtered Approximate Nearest Neighbor Search}

\author{Haoyu Wang}
\affiliation{%
  \institution{Xidian University}
  \city{Xi'an}
  \country{China}
}
\email{hyw2004@stu.xidian.edu.cn}

\author{Yandi Zhang}
\orcid{0000-0002-1825-0097}
\affiliation{%
  \institution{Xidian University}
  \city{Xi'an}
  \country{China}
}
\email{ydzhang\_2@stu.xidian.edu.cn}

\author{Jiadong Xie}
\orcid{0000-0001-5109-3700}
\affiliation{%
  \institution{The Chinese University of Hong Kong}
  \city{Hong Kong SAR}
  \country{China}
}
\email{jdxie@se.cuhk.edu.hk}

\author{Yingfan Liu}
\authornote{Yingfan Liu is the corresponding author.}
\affiliation{%
  \institution{Xidian University}
  \city{Xi'an}
  \country{China}
}
\email{liuyingfan@xidian.edu.cn}

\author{Hui Li}
\affiliation{%
  \institution{Xidian University}
  \city{Xi'an}
  \country{China}
}
\email{hli@xidian.edu.cn}

\author{Jeffrey Xu Yu}
\affiliation{%
  \institution{The Hong Kong University of Science and Technology (Guangzhou)}
  \city{Guangzhou}
  \country{China}
}
\email{jeffreyxuyu@hkust-gz.edu.cn}

\author{Jiangtao Cui}
\affiliation{%
  \institution{Xi'an University of Posts and Telecommunications}
  \city{Xi'an}
  \country{China}
}
\affiliation{%
  \institution{Xidian University}
  \city{Xi'an}
  \country{China}
}
\email{cuijt@xupt.edu.cn}

\begin{abstract}
Range-filtered approximate nearest neighbor search (RFANNS) is an important primitive for vector databases; it retrieves vectors that are similar to a query and satisfy a numerical range predicate, but existing RFANNS indexes expose vectors, attributes, and queries in plaintext.
This assumption is unsuitable for outsourced vector databases, where sensitive data and queries must be protected from an honest-but-curious cloud server.
To the best of our knowledge, this is the first study that systematically formulates and evaluates privacy-preserving RFANNS over outsourced encrypted vector databases.
Our approach separates range localization from encrypted vector search: an authorized user maps the query range to a compact set of nodes in a local N-ary attribute tree, and the server searches only the corresponding proximity graph sub-indices over encrypted vectors. To reduce expensive encrypted comparisons, we use a filter-and-refine pipeline that first retrieves coarse candidates with approximate distance-comparison-preserving encryption and then reranks a small candidate set with exact distance-comparison encryption. We then analyze the computation, storage, communication, and leakage of the protocol. Experiments on four widely used vector datasets show that our method improves the QPS–Recall trade-off over representative secure adaptations of existing RFANNS approaches, scaling effectively to large datasets.
%

\end{abstract}

\maketitle

\pagestyle{\vldbpagestyle}
\begingroup\small\noindent\raggedright\textbf{PVLDB Reference Format:}\\
\vldbauthors. \vldbtitle. PVLDB, \vldbvolume(\vldbissue): \vldbpages, \vldbyear.\\
\href{https://doi.org/\vldbdoi}{doi:\vldbdoi}
\endgroup
\begingroup
\renewcommand\thefootnote{}\footnote{\noindent
This work is licensed under the Creative Commons BY-NC-ND 4.0 International License. Visit \url{https://creativecommons.org/licenses/by-nc-nd/4.0/} to view a copy of this license. For any use beyond those covered by this license, obtain permission by emailing \href{mailto:info@vldb.org}{info@vldb.org}. Copyright is held by the owner/author(s). Publication rights licensed to the VLDB Endowment. \\
\raggedright Proceedings of the VLDB Endowment, Vol. \vldbvolume, No. \vldbissue\ %
ISSN 2150-8097. \\
\href{https://doi.org/\vldbdoi}{doi:\vldbdoi} \\
}\addtocounter{footnote}{-1}\endgroup

\ifdefempty{\vldbavailabilityurl}{}{
\vspace{.3cm}
\begingroup\small\noindent\raggedright\textbf{PVLDB Artifact Availability:}\\
The source code, data, and/or other artifacts have been made available at \url{\vldbavailabilityurl}.
\endgroup
}


\section{Introduction}

With the increasing adoption of vector representations, objects such as images~\cite{radford2021learning, he2022masked} and texts~\cite{feng2022language,wang2024improving} are encoded as high-dimensional vectors that capture their semantic information. Vector similarity search is commonly supported by $k$-approximate nearest neighbor search ($k$-ANNS), which has become a fundamental component of vector databases and retrieval-augmented generation~\cite{RAG,RAG-ACL}.
In many applications, pure $k$-ANNS is insufficient because the returned objects must also satisfy structured constraints. For example, an e-commerce user may search for products visually similar to a query image within a specified price range, while a multimedia application may retrieve similar objects appearing within a temporal interval. These scenarios motivate \emph{range-filtered approximate nearest neighbor search} (RFANNS), where each object consists of a vector and a numerical attribute, and each query contains a query vector and a range predicate over the attribute domain.

Recent plaintext RFANNS methods~\cite{WST,iRange,JiangYZHSZLW25,wang2025wow} combine vector indices such as HNSW~\cite{Hnsw} with range-aware structures such as segment trees and B-trees, achieving efficient range-constrained vector search. However, they assume that vectors, attributes, and queries are available in plaintext. This assumption is unsuitable for outsourced data management, where data owners rely on cloud servers for storage and query processing but wish to protect database vectors, numerical attributes, query vectors, and query ranges.
This motivates \emph{privacy-preserving range-filtered approximate nearest neighbor search} (PP-RFANNS). In this setting, a data owner outsources an encrypted vector database to an honest-but-curious cloud server, and authorized users issue encrypted RFANNS queries. The server should return high-quality approximate neighbors satisfying the range predicate, without learning the protected database and query contents beyond explicitly characterized leakage.

Supporting PP-RFANNS is challenging because it requires simultaneously satisfying three conflicting goals: security, accuracy, and efficiency, because encryption usually destroys the direct usability of conventional vector indices and attribute filters. In particular, exact secure distance comparison~\cite{zheng2024achieving,liu2025ppanns} is often much more expensive than plaintext distance computation, while secure range checking over encrypted attributes can introduce significant cryptographic overhead or require additional interaction~\cite{zhang2024performance}.
Furthermore, existing approaches do not directly resolve this tension.
Plaintext RFANNS methods violate the privacy requirement, while privacy-preserving ANNS methods~\cite{peng2017reusable,xue2018secure,chen2020sanns,servan2022private,zhou2025pacmann,liu2025ppanns} support only pure vector search and do not efficiently handle range predicates. 
Straightforward secure adaptations also suffer from selectivity-sensitive bottlenecks: secure pre-filtering performs expensive encrypted vector search over all qualified objects, secure post-filtering requires substantial over-fetching under low selectivity, and encrypted adaptations of dedicated RFANNS indices such as iRangeGraph~\cite{iRange} repeatedly invoke costly secure range checks as new neighbors are inspected during online graph traversal.

To address these limitations, we propose a practical PP-RFANNS scheme that decouples range localization from encrypted vector search. Our scheme employs an N-ary tree–HNSW hybrid index. The N-ary tree, maintained by the query user, hierarchically organizes the attribute domain, while the cloud server stores an HNSW sub-index for each tree node over the encrypted vectors in its interval. Given a query, the query user locally maps the range predicate to a compact set of matched tree nodes, after which the cloud server searches only the corresponding encrypted HNSW sub-indices.
We further design a filter-and-refine query strategy. During filtering, the server uses DCPE ciphertexts~\cite{fuchsbauer2022approximate}, which approximately preserve distance ordering, to efficiently retrieve coarse candidates from the matched sub-indices. It then merges these candidates and applies DCE-based exact distance comparison~\cite{liu2025ppanns} only to a reduced candidate set. This design avoids online encrypted range checking and limits expensive exact comparisons while preserving high search accuracy. Moreover, we analyze the computation, communication, storage, and security properties of the proposed scheme under the honest-but-curious model, and explicitly characterize the leakage arising from the encrypted vector representations, index structures, and query execution. Experiments on widely used vector datasets show that our method achieves consistently better efficiency–accuracy trade-offs than representative secure adaptations of existing RFANNS approaches across different query selectivities, while scaling effectively to larger datasets. {At $Recall@10$ = 0.95, PP-RFANNS improves query throughput by at least two orders of magnitude over the secure baselines.} 


Our main contributions are summarized as follows.
\ding{202} We formulate the PP-RFANNS problem over outsourced encrypted data and identify the security, accuracy, and efficiency challenges arising from the joint processing of vector similarity and range predicates.
\ding{203} We propose an N-ary tree–HNSW hybrid index that separates local range localization from server-side encrypted vector search, allowing the server to search only the sub-indices relevant to the query range.
\ding{204} We design a filter-and-refine query strategy that combines efficient DCPE-based candidate retrieval with DCE-based exact refinement, balancing query efficiency and search accuracy.
\ding{205} We provide cost and security analyses and conduct extensive experiments demonstrating that the proposed scheme consistently outperforms representative secure adaptations under different settings.

\section{Preliminaries}

In this section, we present the system model and threat model, followed by our problem definition and design goals. 

\subsection{System Model}
As shown in Figure~\ref{fig:system_model}, the system consists of three participants: a data owner (DO), a query user (QU), and a cloud server (CS).

\begin{figure}[t]
  \centering
  \includegraphics[width=\linewidth]{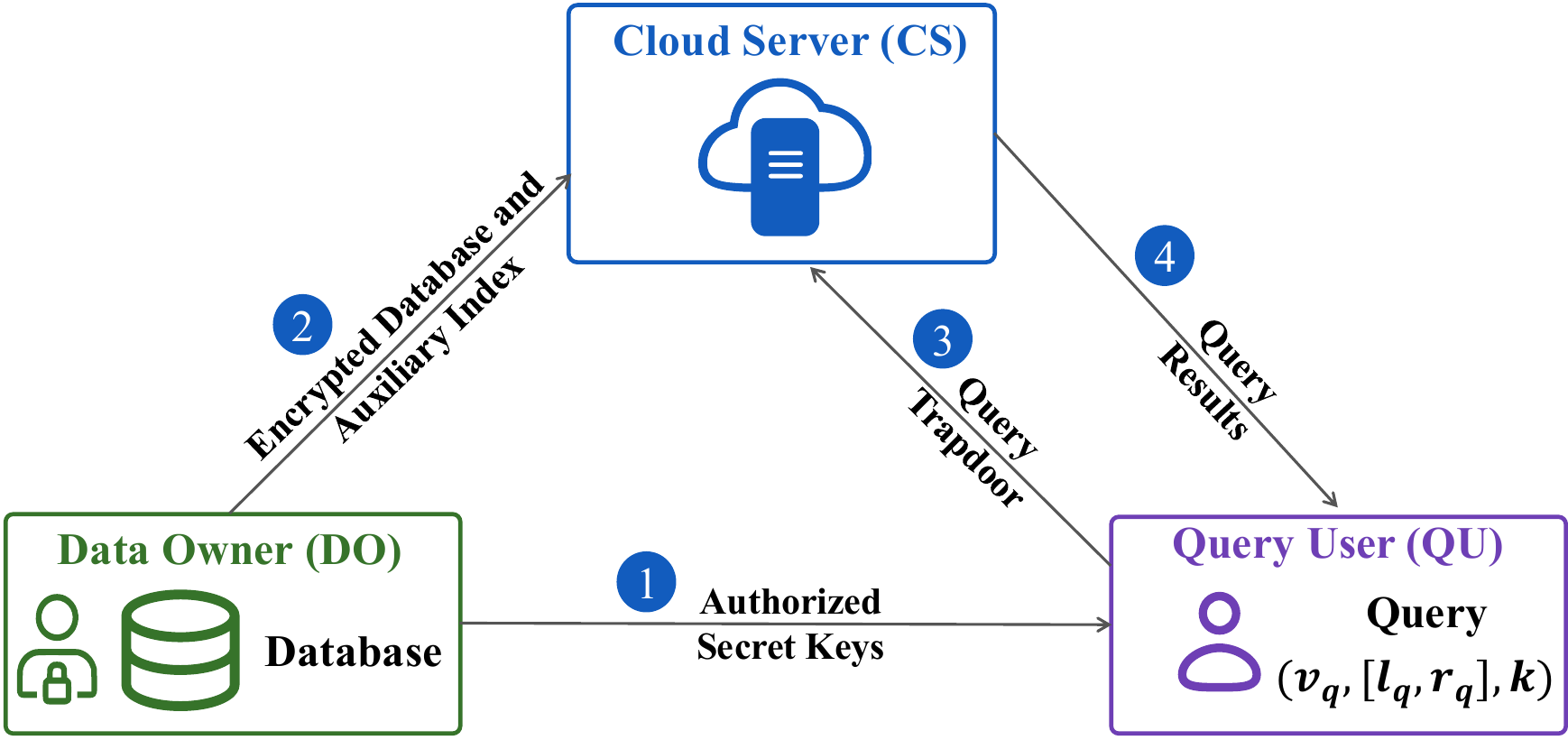}
  \caption{The system model of PP-RFANNS.}
  \label{fig:system_model}
\end{figure}

\noindent\textbf{Data Owner (DO): } The DO holds a database of objects, denoted as $O = \{o_i\}_{i=1}^n$, where each $o_i = (v_i, a_i)$ consists of a $d$-dimensional vector $v_i \in \mathbb{R}^d$ and a numerical attribute {$a_i \in \mathcal{A} \subset \mathbb{Z}$}~\footnote{This setting can be naturally extended to $a_i \in \mathbb{R}$ through \emph{IEEE 754 sortable integer encoding}, where $\mathcal{A}$ represents the attribute domain. Specifically, floating-point values can be transformed into order-preserving integer keys, such that their numerical order is preserved after encoding. As a result, range queries over real-valued attributes can be equivalently processed in the encoded integer domain.}. 
The DO encrypts the objects into ciphertexts and outsources them to the CS. To support efficient PP-RFANNS queries, the DO builds a privacy-preserving index and distributes it to the CS. Besides, the DO sends the secret keys for trapdoor generation to the authorized QUs. 
\\
\textbf{Query User (QU): } 
An RFANNS query $q = (v_q,[l_q,r_q],k)$ includes three parts, i.e., (1) the query vector $v_q\in\mathbb{R}^d$, the query range $[l_q,r_q] \subset \mathcal{A}$ and the number $k$ of returned neighbors. 
To protect the query $q$, the QU encrypts $q$ into the query trapdoor $T_q$ according to the secret keys, and then sends $T_q$ to the CS.
\\
\textbf{Cloud Server (CS): }
The CS stores the encrypted objects and the auxiliary index from the DO to answer efficient PP-RFANNS queries. Once the query trapdoor $T_q$ is received from the QU, it conducts the search process over the encrypted objects and the index and returns the results to the QU. 

\subsection{Threat Model}

We consider an outsourced PP-RFANNS setting involving a DO, authorized QUs, and one or more CSs that store encrypted data and participate in query processing. The DO and authorized QUs are assumed to be honest and follow the prescribed procedures. In particular, they do not disclose secret keys, plaintext objects, or plaintext queries to the CSs or any unauthorized party.
The CSs are assumed to be \emph{honest-but-curious}. They faithfully store the outsourced encrypted objects, maintain the required server-side indices, and execute the prescribed query-processing protocols, but may attempt to infer sensitive information from their respective views. Such views may include encrypted objects, query trapdoors, server-side index structures, intermediate protocol messages, and query-execution traces. If a PP-RFANNS scheme employs multiple CSs, its assumptions regarding possible collusion among them must be explicitly specified by that scheme.
Encrypted indexing and query processing may inevitably expose scheme-dependent structural, access-pattern, and cryptographic information, such as index topology, accessed objects, range-checking outcomes, and distance-comparison information. The precise leakage depends on the adopted scheme. We characterize the leakage of our proposed scheme in Section~\ref{sec:security}.

\subsection{Problem Definition}

Given an object set $O$, let $V = \{v_i | o_i \in O\}$ be the vector set and $A = \{a_i | o_i \in O\}$ be the set of attribute values. For a query $q = (v_q, [l_q, r_q], k)$, let $O_{[l_q, r_q]} = \{o_i | a_i \in [l_q, r_q]\}$ be the set of qualified objects and $V_{[l_q, r_q]} = \{v_i | o_i \in O_{[l_q, r_q]}\}$ be the qualified vector set. 

\begin{definition}[$k$-ANNS]
\label{def:k-anns}
Given a vector set $V \subset \mathbb{R}^d$, a query vector $v_q \in \mathbb{R}^d$ and an integer $k$, $k$-ANNS returns $k$ sufficiently close neighbors in $V$ to $v_q$. 
\end{definition} 

As the underlying search primitive of RFANNS and PP-RFANNS, the k-ANNS problem has been studied for a few decades. According to recent studies~\cite{DPG, survey2021, AziziEP25}, proximity graphs (PG) such as HNSW~\cite{Hnsw}, NSG~\cite{Nsg} and ALMG~\cite{ALMG} have been the SOTA $k$-ANNS methods. A PG treats each vector in $V$ as a graph vertex and builds edges between each node and its closest neighbors in $V$. Different PGs share the same vertex set but have different edge sets due to various edge selection strategies. Notably, PG is the building block of the mainstream RFANNS methods and our PP-RFANNS scheme in this work. 

\begin{definition}[RFANNS]
\label{def:rfanns}
Given an object set $O$ and a query $q$, RFANNS returns the $k$-ANN of $v_q$ over the vector set $V_{[l_q, r_q]}$. 
\end{definition}



\begin{definition}[PP-RFANNS]
\label{def:pp-rfanns}
Given an encrypted representation $C_O$ of $O$ and a query trapdoor $T_q$ generated for a query $q$, PP-RFANNS returns the identifiers of $k$-ANN of $v_q$ from the set $V_{[l_q,r_q]}$ by operating over encrypted data and encrypted queries.
\end{definition}

The security of a PP-RFANNS scheme depends on both the guarantees of its underlying cryptographic primitives and the information exposed by its encrypted index and query execution. In this work, we do not claim a new simulation-based security result for the complete composition. Instead, we rely on the established security properties of the adopted primitives and explicitly characterize the additional leakage of the complete protocol in Section~\ref{sec:security}.


\subsection{Design Goals}
Our goal is to design a practical PP-RFANNS scheme satisfying the following requirements.
\begin{itemize}
    \item \textbf{Security.} 
    The scheme should protect the plaintext database vectors and numerical attributes, as well as the query vectors, numerical query bounds, and exact distance values, against an honest-but-curious CS, except for the information explicitly revealed by the underlying encryption primitives, the server-side index, and query execution. 
    \item \textbf{Accuracy.} The scheme should return high-quality $k$-ANN of the query vector $v_q$ from the qualified vector set $V_{[l_q,r_q]}$.
    \item \textbf{Efficiency.} The scheme should support efficient PP-RFANNS query processing at the CS, while incurring low computation overhead at the QU and low communication overhead between the QU and the CS.
\end{itemize}

\section{Limitations of Existing Works}
\label{sec:limit}

In this section, we first review plaintext RFANNS methods, privacy-preserving range query and privacy-preserving $k$-ANNS (PP-ANNS) methods, and then discuss the secure adaptations of existing RFANNS methods to answer PP-RFANNS. 

\subsection{Plaintext RFANNS Methods}

Existing RFANNS methods can be divided into general-purpose approaches and dedicated methods. General-purpose approaches include pre-filtering~\cite{ADBV,milvus,vbase}, post-filtering~\cite{ADBV,milvus}, and ACORN~\cite{ACORN}. Pre-filtering performs brute-force search over the objects satisfying the range predicate, whereas post-filtering first conducts global $k$-ANNS and then removes unqualified results. ACORN modifies HNSW~\cite{Hnsw} by enlarging node degrees and restricting traversal to qualified nodes.
Dedicated RFANNS methods~\cite{SeRF,WST,iRange,JiangYZHSZLW25,wang2025wow,MSTG} combine range-aware structures, such as segment trees~\cite{WST,iRange} or B-trees~\cite{JiangYZHSZLW25}, with PG indexes. By restricting graph search to range-relevant partitions or edges, they achieve substantially better query performance than general-purpose approaches.

Among the dedicated methods, iRangeGraph~\cite{iRange} builds a segment-tree-based multi-layer index, where each tree node maintains a PG index, called an \emph{elemental graph}, over the objects in its attribute interval. During query processing, it dynamically selects in-range edges from the relevant elemental graphs to build a range-dedicated graph during traversal. This online edge-selection procedure requires range checks over candidates.


\underline{\textbf{Issue:}} These methods solve RFANNS by exploiting plaintext data, which should be protected in our privacy-preserving scenarios, and violate the goal of security. Hence, they cannot solve PP-RFANNS. 

\subsection{Privacy-Preserving Range Query}


Privacy-preserving range queries retrieve encrypted records whose numerical attributes fall within a query interval $[l_q,r_q]$, while protecting the attribute values and query bounds from the cloud server. Existing solutions mainly rely on Order-Preserving/Order-Revealing Encryption (OPE/ORE), Searchable Symmetric Encryption (SSE), or Symmetric Homomorphic Encryption (SHE).

OPE/ORE-based methods~\cite{boldyreva2009order,lewi2016order,kerschbaum2015frequency} and SSE-based methods~\cite{demertzis2016practical,wang2019forward,molla2023efficient} support efficient range retrieval through encrypted indices, but reveal information such as ciphertext order, search patterns, access patterns, or result sizes.
SHE-based methods~\cite{mahdikhani2020achieving,zheng2021efficient,zhang2024performance} evaluate range predicates directly over encrypted attributes and avoid explicitly revealing their global order. A basic construction evaluates the encrypted range predicate independently for each examined record. 
Existing schemes reduce this cost through auxiliary index structures, which limit the number of examined records, and optimized SHE constructions, such as iSHE~\cite{zhang2024performance}, which improve the efficiency of each encrypted predicate evaluation. Nevertheless, repeated range checks can still incur substantial cryptographic
computation and, when decryption is delegated to a non-colluding auxiliary server, inter-server communication and interaction overhead.

\underline{\textbf{Issue:}} These methods focus on secure range query over encrypted attributes, but do not support high-dimensional encrypted vector search. Hence, they cannot directly solve PP-RFANNS. 

\subsection{Privacy-Preserving ANNS Methods}

Existing privacy-preserving ANNS (PP-ANNS) methods protect outsourced vectors and auxiliary ANN indices, while using encrypted indices to reduce the search space. According to whether vector distances can be evaluated directly over ciphertexts, they can be broadly classified into distance-incomparable and distance-comparable approaches.

Distance-incomparable approaches, including AES- and PIR-based schemes~\cite{peng2017reusable,servan2022private,zhou2025pacmann}, cannot support distance-guided search at the CS. Although an auxiliary index may identify a candidate subset, the QU must retrieve and decrypt candidates or interactively access encrypted index entries, resulting in substantial communication and QU-side computation.
Distance-comparable methods employ primitives such as ASPE~\cite{wong2009secure,li2019insecurity,miao2023efficient,yuan2017practical}, DCPE~\cite{fuchsbauer2022approximate}, AME~\cite{zheng2024achieving}, HE~\cite{furon2013fast,xue2018secure}, or DCE~\cite{liu2025ppanns} to enable distance-guided search over encrypted vectors. However, indices built from plaintext vectors may reveal proximity relationships, while secure distance evaluation at every search step can be expensive. Moreover, ASPE-based schemes have known security weaknesses~\cite{lin2017revisiting,li2019insecurity,liu2025ppanns}, DCPE preserves only approximate distance order, and AME- and HE-based methods incur relatively high computational overhead.
Liu~\textit{et al.}~\cite{liu2025ppanns} address these issues by constructing HNSW directly over DCPE-encrypted vectors. The CS first uses DCPE for efficient coarse retrieval and then applies DCE only to the reduced candidate set for exact reranking. This design supports efficient encrypted graph search while limiting the number of expensive exact distance comparisons.

\underline{\textbf{Issue:}} They focus on pure ANNS without handling attribute values, and thus cannot directly support efficient PP-RFANNS.

\subsection{Secure Adaptations of RFANNS}

\begin{figure}[t]
  \centering
  \includegraphics[width=\linewidth]{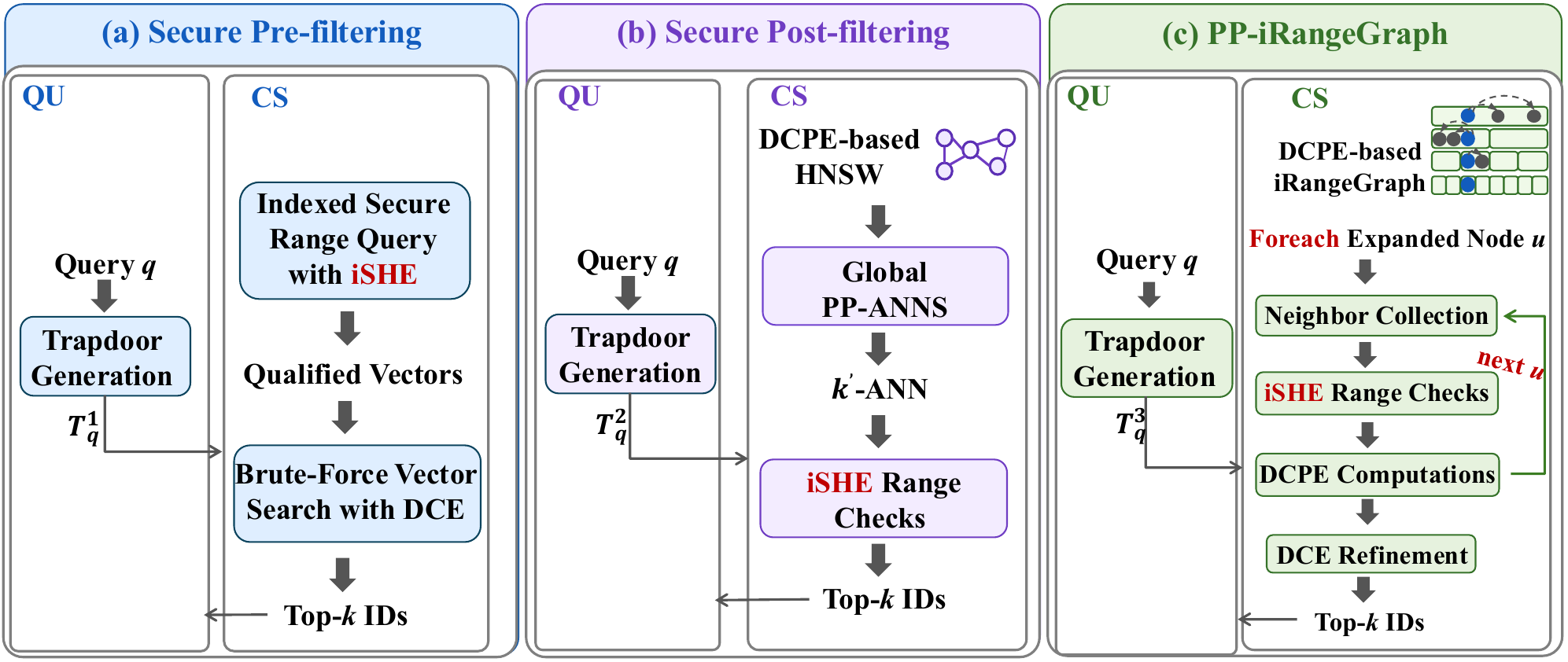}
  \caption{The illustrations of secure adaptations}
  \label{fig:secure_adaptations}
\end{figure}

We consider simple adaptations of three representative RFANNS methods in the encrypted setting, denoted as secure pre-filtering, secure post-filtering, and PP-iRangeGraph, as illustrated in Figure~\ref{fig:secure_adaptations}. 
We instantiate all three secure adaptations with iSHE-based range checking. iSHE naturally supports candidate-wise range checks, which fits Secure Post-filtering and PP-iRangeGraph. OPE/ORE-based alternatives reveal the global order of encrypted attributes, while range-SSE introduces a different encrypted-index leakage profile. Using the same iSHE-based protocol provides a consistent basis for
comparing the three query-processing strategies.

\noindent\underline{\textbf{Secure Pre-filtering.}}
Given a query $q=(v_q,[l_q,r_q],k)$, the QU generates the query trapdoor and sends it to the CS. The CS searches a B-tree over encrypted attributes and selectively invokes iSHE, to identify the qualified vectors. It then performs brute-force DCE-based distance comparisons over all qualified vectors and returns the top-$k$ results.
Let $m_q^{\mathsf{pre}}$ be the number of iSHE checks and $s_q$ the query selectivity. The query cost includes $m_q^{\mathsf{pre}}C_{\mathsf{rc}}$ for range checking and $O(s_q n\log k)$ DCE comparisons. Although the B-tree reduces range checks, the exhaustive encrypted vector search becomes expensive at high selectivity.

\noindent\underline{\textbf{Secure Post-filtering.}}
Secure Post-filtering first performs global PP-ANNS~\cite{liu2025ppanns} over the entire encrypted vector database. Specifically, the CS searches a DCPE-based HNSW index and retrieves the global top-$k'$ candidates, where $k' \geq k$. {It then performs iSHE-based range checks over these candidates, and applies DCE refinement to the qualified candidates to obtain the final top-$k$ results.}
The main difficulty is selecting $k'$. Under low selectivity, a small $k'$ may return fewer than $k$ qualified results and require repeated searches, whereas a large $k'$ incurs unnecessary encrypted vector search and iSHE-based range-checking overhead. Consequently, Secure Post-filtering performs poorly for low-selectivity queries.

\noindent\underline{\textbf{PP-iRangeGraph.}}
PP-iRangeGraph adapts iRangeGraph by encrypting both the vectors and numerical attributes. 
The CS iteratively searches the DCPE-based elemental graphs. In each iteration, it selects a node $u$ for expansion and collects its candidate neighbors from the relevant elemental graphs. For every inspected neighbor, the CS invokes iSHE to determine whether its encrypted attribute satisfies the query predicate. {Qualified neighbors are then used to update the DCPE-based coarse candidate set. After the graph search terminates, the CS applies DCE refinement to the DCPE-based coarse candidates and returns the top-$k$ results.}
{If $m_q$ neighbors are inspected, PP-iRangeGraph performs $m_q$ iSHE-based range checks during graph traversal, followed by DCE refinement on the final candidate set.} Repeated iSHE invocation during graph traversal incurs substantial cryptographic computation, inter-server communication, and interaction overhead.



\noindent\underline{\textbf{Issue:}}
In the secure adaptations, all three baselines use iSHE-based range checking. Under our threat model, the resulting ciphertexts are decrypted by a non-colluding auxiliary server, which introduces substantial cryptographic computation, communication, and interaction overhead. Beyond this common cost, the three methods exhibit different selectivity-sensitive bottlenecks. Secure Pre-filtering reduces the number of range checks through an encrypted range index, but still performs exhaustive DCE-based distance comparison over all qualified objects and therefore becomes expensive at high selectivity. Secure Post-filtering applies range checks to candidates returned by global PP-ANNS and may require substantial over-fetching or repeated searches at low selectivity. PP-iRangeGraph repeatedly invokes range checking as new neighbors are inspected during online graph traversal, causing persistent cryptographic and interaction overhead. Consequently, none of these adaptations provides consistently efficient PP-RFANNS query processing across different selectivities.

\section{Our PP-RFANNS Scheme}




In this section, we present our PP-RFANNS scheme. The key idea is to decouple range localization from encrypted vector search, such that the CS searches only the secure vector sub-indices corresponding to the query range. We first give an overview of the scheme, and then describe secure index construction, secure query processing, and cost analysis.

\subsection{Overview}

Our proposed PP-RFANNS scheme consists of two main procedures: \textbf{secure index construction} and \textbf{secure query processing}, as shown in Figure~\ref{fig:index_construction} and Figure~\ref{fig:query_process}, respectively. Given an object set $O=(V,A)$, the DO constructs the secure index in three phases. First, it encrypts the vector set $V$ using DCPE and DCE, producing $C_V^{\mathsf{DCPE}}$ and $C_V^{\mathsf{DCE}}$, respectively. Second, it constructs a range-aware N-ary tree $\mathcal{T}$ over the attribute set $A$. Third, it builds an HNSW sub-index $G_{tn}$ for each tree node $tn \in \mathcal{T}$. After index construction, the DO distributes the DCPE and DCE secret keys, denoted by $SK_{\mathsf{DCPE}}$ and $SK_{\mathsf{DCE}}$, together with $\mathcal{T}$, to the QU. Meanwhile, $C_V^{\mathsf{DCPE}}$, $C_V^{\mathsf{DCE}}$, and the HNSW sub-indices are outsourced to the CS.

During secure query processing, the QU and the CS perform distinct tasks: the QU locally interprets the range predicate and generates encrypted query trapdoors, whereas the CS executes the vector search over the encrypted sub-indices selected by the QU. Specifically, for a query $q=(v_q,[l_q,r_q],k)$, the QU generates the query trapdoor $T_q=(T_q^{\mathsf{DCPE}},T_q^{\mathsf{DCE}})$ using $SK_{\mathsf{DCPE}}$ and $SK_{\mathsf{DCE}}$, respectively. It then evaluates the range predicate over $\mathcal{T}$ to identify the set $I_q$ of matching tree nodes and sends $(T_q,I_q)$ to the CS. Upon receiving them, the CS employs a filter-and-refine strategy to answer the query.


\begin{figure}
  \centering
  \includegraphics[width=\linewidth]{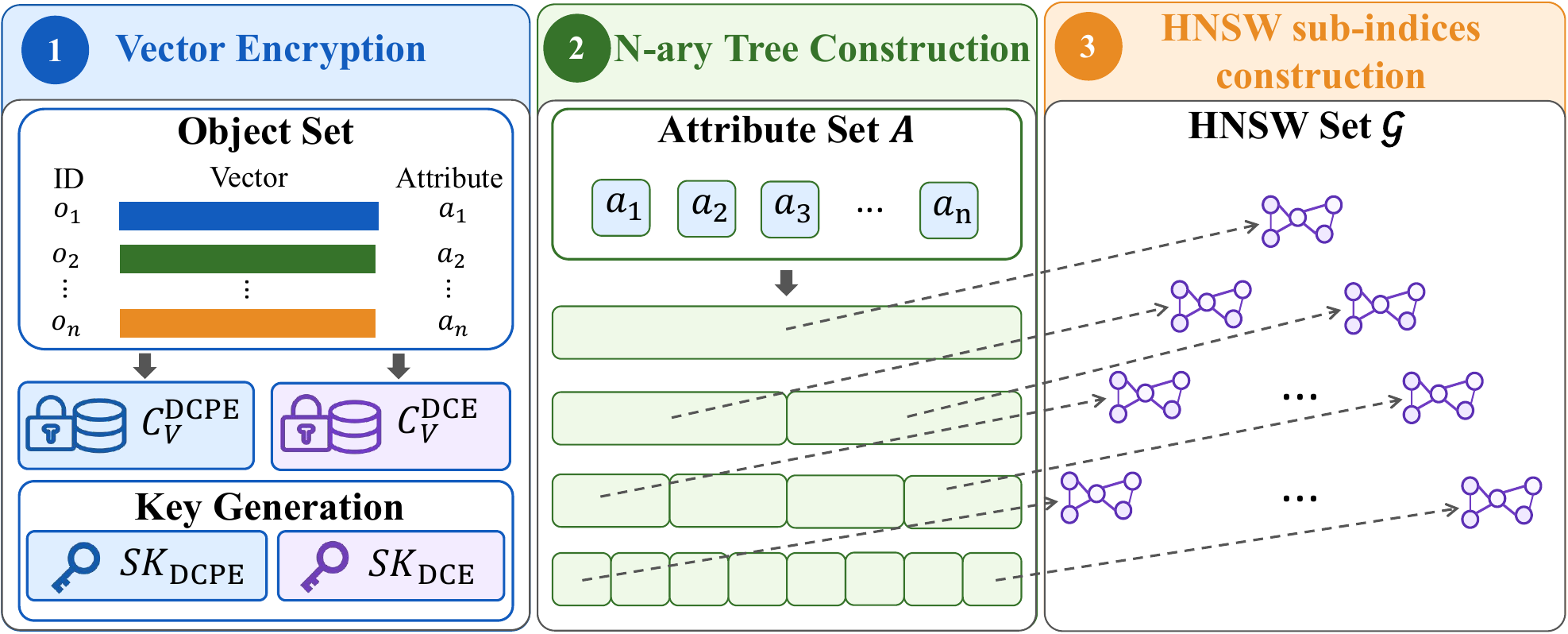}
  \caption{The workflow of secure index construction}
  \label{fig:index_construction}
\end{figure}

\begin{figure}
  \centering
  \includegraphics[width=\linewidth]{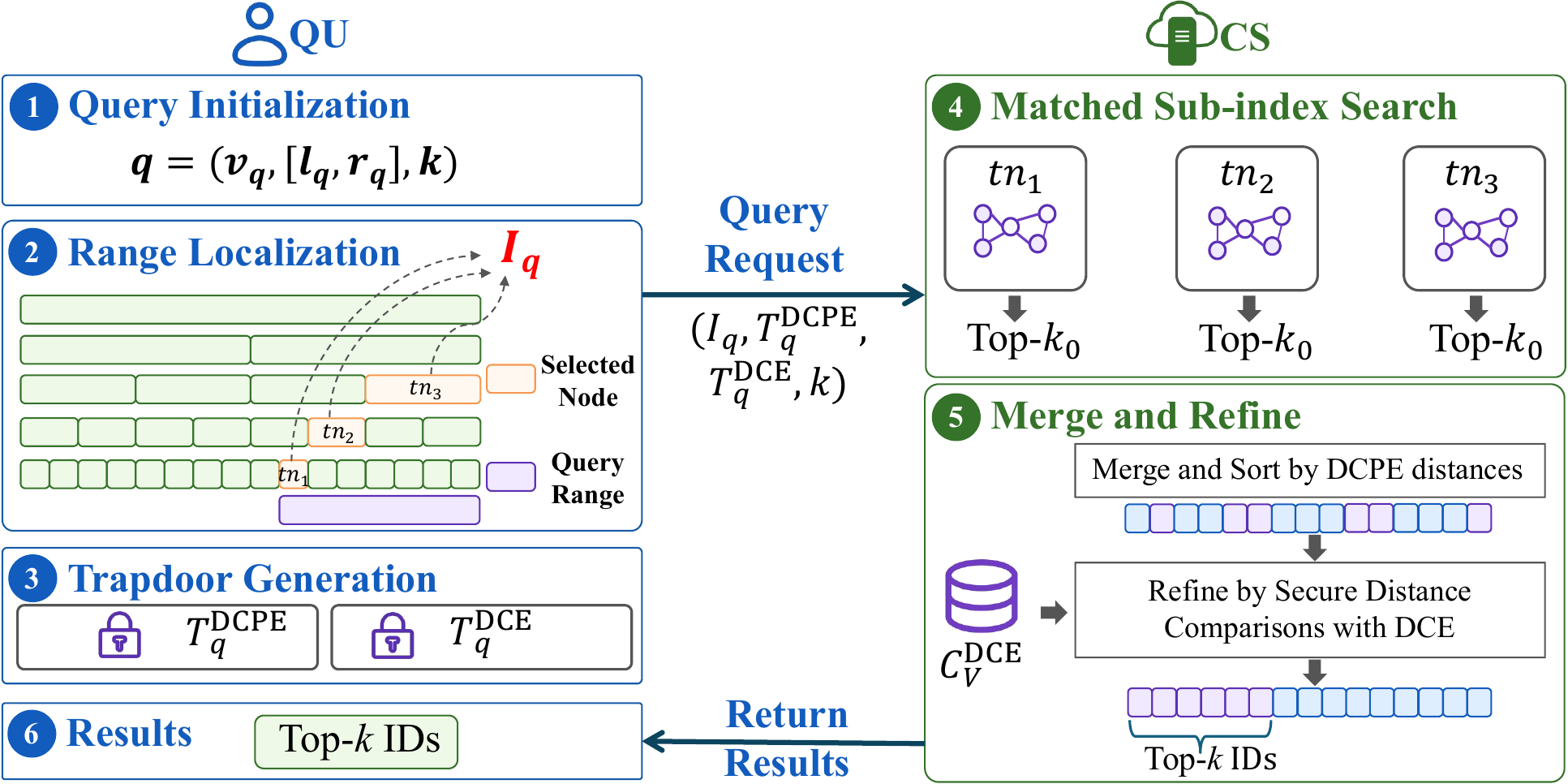}
  \caption{The workflow of secure query processing}
  \label{fig:query_process}
\end{figure}






\subsection{Secure Index Construction}

As in Figure~\ref{fig:index_construction}, the secure index construction consists of three phases, i.e., \emph{vector encryption}, \emph{N-ary tree construction} and \emph{HNSW sub-indices construction}. 


\noindent\underline{\textbf{Vector Encryption.}}
The DO encrypts each vector $v\in V$ using both DCPE and DCE, which serve different purposes in our scheme. The DCPE ciphertexts are used to construct and search the HNSW sub-indices, while the DCE ciphertexts are used for exact secure distance comparisons.

For DCPE, we adopt the Scale-and-Perturb (SAP) construction in~\cite{fuchsbauer2022approximate}. Given the secret key $SK_{\mathsf{DCPE}}=(s,\beta)$, the DO generates the ciphertext as $C_v^{\mathsf{DCPE}}=s\cdot v+\lambda_v$, where $s\in\mathbb{R}^{+}$ is a random scaling factor and $\lambda_v$ is a random vector sampled from the $d$-dimensional ball centered at the origin with radius $s\beta/4$. The resulting ciphertext $C_v^{\mathsf{DCPE}}\in\mathbb{R}^{d}$ approximately preserves the distances among vectors. Encrypting one vector requires $O(d)$ time and $O(d)$ space. 

For DCE, the DO first generates the secret key $SK_{\mathsf{DCE}}=\{M_1,M_2,\\M_3,\pi_1,\pi_2,r_1,r_2,r_3,r_4,kv_1,kv_2,kv_3,kv_4\}$. Here, $M_1$, $M_2$, and $M_3$ are random invertible matrices, $\pi_1$ and $\pi_2$ are random permutation functions, $r_1,\ldots,r_4$ are random scalars, and $kv_1,\ldots,kv_4$ are random masking vectors. The DO then computes the DCE ciphertext as $C_v^{\mathsf{DCE}}\leftarrow\mathsf{DCE.Enc}(v,SK_{\mathsf{DCE}})$. The encryption consists of two main steps. First, vector randomization permutes, splits, and pads $v$ with random values, producing an intermediate vector $\bar{v}\in\mathbb{R}^{d+8}$. Second, vector transformation applies random matrix transformations and element-wise masking to $\bar{v}$, producing the ciphertext $C_v^{\mathsf{DCE}}=(\bar{v}'_1,\bar{v}'_2,\bar{v}'_3,\bar{v}'_4)$, where each component satisfies $\bar{v}'_j\in\mathbb{R}^{2d+16}$. Therefore, each DCE ciphertext contains $8d+64$ scalar values. Since the encryption involves matrix--vector multiplications over matrices whose dimensions are linear in $d$, its computational complexity is $O(d^2)$, while its space complexity is $O(d)$. 

After encrypting all vectors, the DO obtains $C_V^{\mathsf{DCPE}}=\{C_v^{\mathsf{DCPE}}\mid v\in V\}$ and $C_V^{\mathsf{DCE}}=\{C_v^{\mathsf{DCE}}\mid v\in V\}$. 

\noindent\underline{\textbf{N-ary Tree Construction.}}
Given the numerical attribute set $A=\{a_1,a_2,\ldots,a_n\}$, the DO constructs a range-aware N-ary tree $\mathcal{T}$ to hierarchically organize $A$. The DO first sorts the objects according to their attribute values and determines the minimum and maximum values, denoted by $a_{\min}$ and $a_{\max}$, respectively. The root node of $\mathcal{T}$ represents the entire attribute domain $[a_{\min},a_{\max}]$ and contains all objects in $O$.
The tree is then constructed recursively in a top-down manner. For each non-leaf node $tn$ associated with an interval $[l_{tn}, r_{tn}] \subset A$ and an object subset $O_{tn} = \{o_i | o_i \in O, a_i \in [l_{tn}, r_{tn}]\}$, the DO partitions $[l_{tn}, r_{tn}]$ into at most $b$ non-overlapping and consecutive sub-intervals,
$[l_{tn}^1, r_{tn}^1],\ldots,[l_{tn}^m, r_{tn}^m]$, where $m \leq b$. The
corresponding object subset is divided as $O_{tn}^{j}=\{o_i\in O_{tn}\mid a_i\in [l_{tn}^j, r_{tn}^j]\}$, and a child node is created for each nonempty subset $O_{tn}^{j}$. The intervals of all child nodes jointly cover the interval of their parent node and satisfy
$[l_{tn}^i, r_{tn}^i]\cap [l_{tn}^j, r_{tn}^j]=\emptyset$ for any $i\neq j$.

To obtain a balanced tree, the partition boundaries are selected according to the sorted attribute values such that the child nodes contain approximately the same number of objects. The recursive partitioning terminates when only one attribute value remains. Consequently, each tree node $tn$ maintains its interval $[l_{tn}, r_{tn}]$ and the identifiers of the objects whose attribute values fall within that interval. The resulting tree $\mathcal{T}$ enables a range predicate to be decomposed into a small set of tree nodes whose intervals are fully contained in the query range.

\noindent\underline{\textbf{HNSW Sub-index Construction.}}
After constructing the N-ary tree $\mathcal{T}$, the DO builds an HNSW sub-index for each tree node in a bottom-up manner. For a tree node $tn$, let $V_{tn}=\{v_i\in V\mid o_i\in O_{tn}\}$ denote its associated vector set, and let $C_{V_{tn}}^{\mathsf{DCPE}}=\{C_{v_i}^{\mathsf{DCPE}}\mid v_i\in V_{tn}\}$ denote the corresponding DCPE ciphertext set. The HNSW sub-index associated with $tn$ is denoted by $G_{tn}$ and indexes $C_{V_{tn}}^{\mathsf{DCPE}}$.

The DO constructs the sub-indices by traversing $\mathcal{T}$ in post-order, such that the sub-indices of all child nodes are available before constructing that of their parent. For each leaf node $tn$, the DO directly builds $G_{tn}$ over $C_{V_{tn}}^{\mathsf{DCPE}}$. In particular, if the leaf contains only one object, $G_{tn}$ consists of a single indexed vector.
Consider an internal node $tn$ with child nodes $\{tn_1,\ldots,tn_m\}$, where $m\leq b$. Since $O_{tn}$ is partitioned by its child nodes, its vector set satisfies $V_{tn}=\bigcup_{j=1}^{m}V_{tn_j}$. To avoid constructing $G_{tn}$ from scratch, the DO selects one child node $tn^\ast$ as the base node and initializes $G_{tn}$ with a copy of $G_{tn^\ast}$. The child with the largest vector set is selected, i.e., $tn^\ast=\arg\max_{tn_j}|V_{tn_j}|$, so as to minimize the number of subsequent insertions. The DO then inserts every DCPE ciphertext in $C_{V_{tn_j}}^{\mathsf{DCPE}}$ for each $tn_j\neq tn^\ast$ into $G_{tn}$ using the native insertion operation of HNSW. Consequently, the resulting sub-index contains exactly the vectors associated with $tn$, i.e., $G_{tn}$ indexes $C_{V_{tn}}^{\mathsf{DCPE}}$.

This procedure is repeatedly applied from the leaf level to the root until an HNSW sub-index has been constructed for every node in $\mathcal{T}$. Finally, the DO obtains the set of HNSW sub-indices $\mathcal{G}=\{G_{tn}\mid tn\in\mathcal{T}\}$.

\noindent\underline{\textbf{Index Outsourcing and Key Distribution.}} After completing the vector encryption and index construction, the DO distributes the generated data to the CS and the QU according to their respective roles. Specifically, the DO sends the set of HNSW sub-indices $\mathcal{G}$ to the CS. Each sub-index is associated with the identifier of its corresponding tree node, allowing the CS to locate the required sub-indices according to the matched node set received from the QU. The DO also sends $C_V^{\mathsf{DCPE}}$ and $C_V^{\mathsf{DCE}}$ to the CS for coarse retrieval and exact secure distance comparisons, respectively. 

In contrast, the DO distributes the N-ary tree $\mathcal{T}$, together with the complete secret keys $SK_{\mathsf{DCPE}}$ and $SK_{\mathsf{DCE}}$, to the QU. The tree $\mathcal{T}$ enables the QU to locally evaluate a range predicate and identify its matched node set, whereas the two secret keys are used to generate the DCPE and DCE trapdoors. The CS receives no component of either secret key. Consequently, the QU holds $\{\mathcal{T},SK_{\mathsf{DCPE}},SK_{\mathsf{DCE}}\}$, while the CS holds $\{\mathcal{G},C_V^{\mathsf{DCPE}}, C_V^{\mathsf{DCE}}\}$ and the necessary auxiliary metadata for secure query processing.

\subsection{Secure Query Processing}

Given a query $q=(v_q,[l_q,r_q],k)$, secure query processing consists of query preparation at the QU and query execution at the CS. The QU locally identifies the tree nodes in $\mathcal{T}$ matching the query range, generates the DCPE and DCE trapdoors for $v_q$, and sends the resulting query request to the CS. Upon receiving the request, the CS adopts a filter-and-refine strategy: it first performs coarse search over the matched HNSW sub-indices to obtain a candidate set, and then applies DCE-based exact distance comparisons to refine the candidates and return the final top-$k$ results.




\begin{figure}
  \centering
  \includegraphics[width=\linewidth]{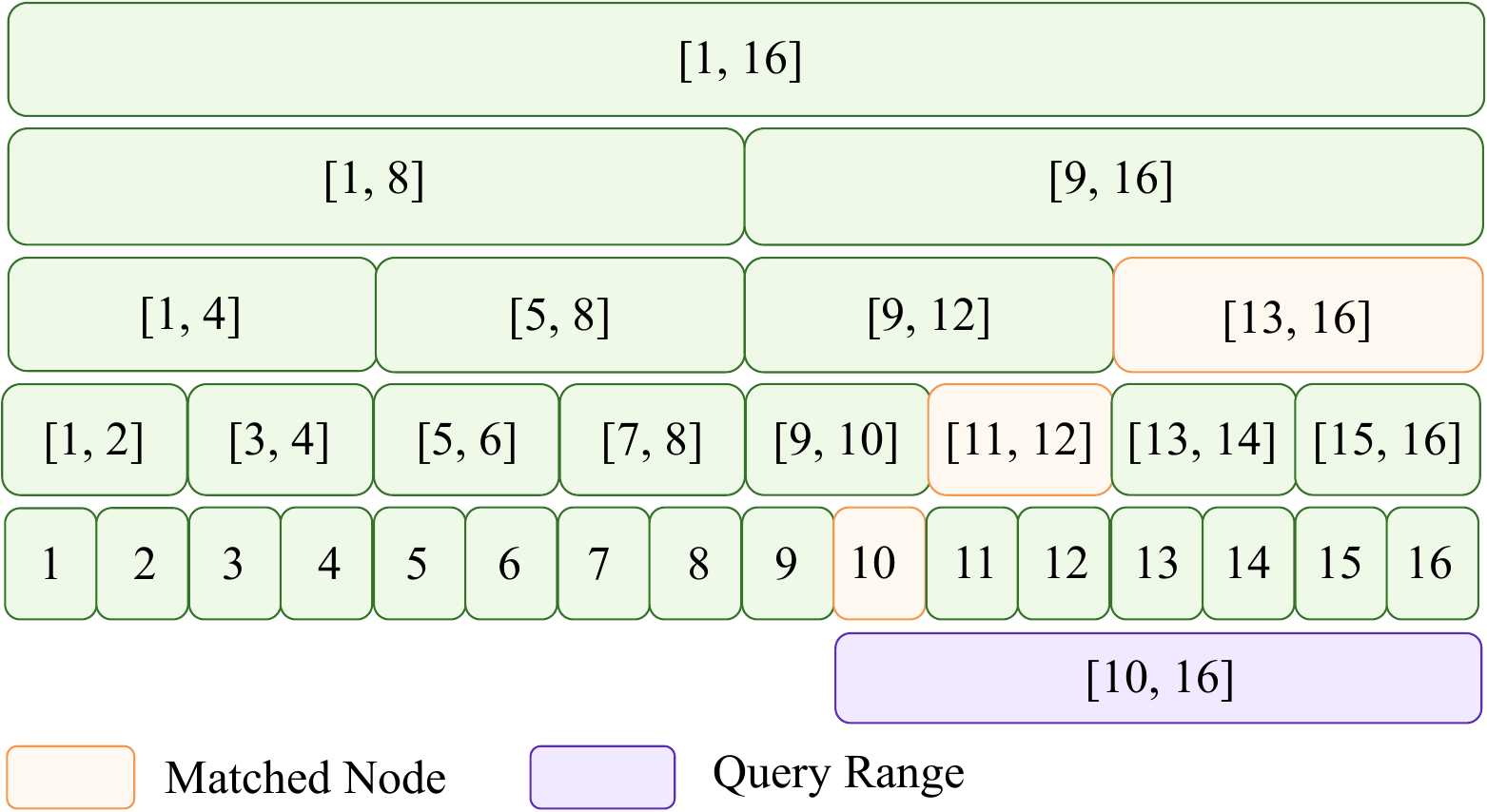}
  \caption{The illustration of range localization over $\mathcal{T}$}
  \label{fig:tree_example}
\end{figure}

\noindent\underline{\textbf{Query Preparation.}}
Given $q$, the QU first performs range localization over its local N-ary tree $\mathcal{T}$. Starting from the root, the QU prunes a node if its interval does not overlap with the query range, selects a node if its interval is fully contained in the query range, and recursively visits its children otherwise. Once a node is selected, its descendants are not further examined. The resulting matched node set, denoted by $I_q$, compactly represents the query range and determines the HNSW sub-indices to be searched by the CS. We illustrate an example of range localization in Figure~\ref{fig:tree_example}.


The QU then encrypts $v_q$ using DCPE and DCE to generate the query trapdoor
$T_q=(T_q^{\mathsf{DCPE}},T_q^{\mathsf{DCE}})$. Finally, it sends $(I_q,T_q,k)$ to the CS for secure query execution.

\begin{theorem}[Exact Range Localization]
\label{thm:exact-range-localization}
Let $I_q$ be the matched node set obtained for the query range $[l_q,r_q]$. For every matched node $tn\in I_q$, its associated interval satisfies $[l_{tn},r_{tn}]\subseteq[l_q,r_q]$. Moreover, the object sets associated with the matched nodes are pairwise disjoint and exactly cover all objects satisfying the range
predicate: $\bigcup_{tn\in I_q} O_{tn} = \{o_i\in O\mid a_i\in[l_q,r_q]\}$. Consequently, every object indexed by the HNSW sub-indices associated with $I_q$ satisfies the query range, and no additional range post-filtering is required at the CS.
\end{theorem}
\begin{proof}
A node is added to $I_q$ only when its interval is fully contained in $[l_q,r_q]$. Hence, for every $tn\in I_q$, $[l_{tn},r_{tn}]\subseteq[l_q,r_q]$, and every object
$o_i\in O_{tn}$ satisfies $a_i\in[l_q,r_q]$.

We next show that every object satisfying the range predicate is covered by a matched node. Consider any object $o_i$ such that $a_i\in[l_q,r_q]$. During the top-down traversal, any node containing $o_i$ cannot be pruned, because its interval contains $a_i$ and therefore overlaps with the query range. If such a node is fully
contained in the query range, it is added to $I_q$ and $o_i$ is covered. Otherwise, the traversal continues along the unique child whose object set contains $o_i$. Since the tree is recursively partitioned until each leaf corresponds to a single attribute value, the traversal eventually reaches a node whose interval is contained in $[l_q,r_q]$. Thus, every qualifying object belongs to the object set of some matched node.

Finally, once a node is added to $I_q$, its descendants are not visited. Therefore, no matched node is an ancestor of another matched node. Since sibling subtrees contain disjoint object sets, the object sets associated with different matched nodes are pairwise disjoint. It follows that $ \bigcup_{tn\in I_q} O_{tn} = \{o_i\in O\mid a_i\in[l_q,r_q]\}$. Therefore, all objects searched by the CS already satisfy the range
predicate, eliminating the need for range post-filtering.
\end{proof}




\noindent\underline{\textbf{CS-side Filtering.}}
Upon receiving the query request $(I_q,T_q,k)$, the CS performs coarse retrieval over the HNSW sub-indices associated with the matched nodes. For each $tn\in I_q$, the CS searches $G_{tn}$ using the DCPE trapdoor $T_q^{\mathsf{DCPE}}$ and retrieves a local candidate list $R_{tn}$ containing at most $k_0$ objects.
Since the local candidate lists are returned from different sub-indices, the CS further merges them into a unified candidate set. Specifically, it uses the distances computed over the DCPE ciphertexts to approximately compare and rank the candidates in
$\bigcup_{tn\in I_q}R_{tn}$. It then retains the top-$k'$ candidates, denoted by $R_q$, where $k'\geq k$ is a refinement parameter.
Since DCPE approximately preserves distance ordering, $R_q$ contains high-quality candidates without revealing the plaintext vectors or their plaintext distances. 



\noindent\underline{\textbf{CS-side Refinement.}} Given the coarse candidate set $R_q$, the CS performs exact secure ranking using the DCE trapdoor $T_q^{\mathsf{DCE}}$. Specifically, it maintains a max-heap $H$ of size at most $k$, which stores the current
top-$k$ results. Each candidate in $R_q$ is inserted into $H$ according to the exact distance ordering obtained through the DCE distance-comparison operation. Once $H$ contains $k$ candidates, a new candidate is retained only if it is closer to $v_q$ than the current farthest candidate in $H$.
After all candidates in $R_q$ have been examined, the objects retained in $H$ form the final top-$k$ result set. Since DCE reveals only the outcomes of distance comparisons, the CS can determine the exact ranking among the candidates without learning their plaintext vectors or exact distance values. Moreover, by Theorem~\ref{thm:exact-range-localization}, no additional range post-filtering is required. The CS finally returns the identifiers of the objects in $H$ to the QU.

\subsection{Cost Analysis}



In this section, we analyze the computation, communication, and storage costs of PP-RFANNS. We consider secure index construction, query processing at the QU and CS, communication between them, and overall space complexity. Let $|I_q|$ denote the number of matched nodes, $k_0$ the number of candidates retrieved from each matched sub-index, and $k’$ the number retained for refinement.



\noindent\underline{\textbf{Index Construction Cost.}} During secure index construction, the DO first encrypts all database vectors using both DCPE and DCE. Since encrypting one vector with DCPE and DCE requires $O(d)$ and $O(d^2)$ time, respectively, the total vector encryption cost is $O(nd^2)$. The DO then builds the N-ary tree $\mathcal{T}$ over $A$, which costs $O(n\log_b |A|)$.
Finally, the DO constructs an HNSW sub-index for each tree node in a bottom-up manner. For each node $tn\in\mathcal{T}$, let $n_{tn}$ be the number of vectors associated with $tn$. Constructing the corresponding HNSW sub-index costs $O(dn_{tn}\log n_{tn})$. Since every vector appears in $O(\log_b |A|)$ tree nodes, we have $\sum_{tn\in\mathcal{T}}n_{tn}=O(n\log_b |A|)$. Moreover, $\log n_{tn}\leq\log n$. Therefore, the total cost of constructing all HNSW sub-indices is $O(dn\log_b |A|\log n)$. Combining the above costs, the overall index construction complexity is $O(nd^2 + n\log_b |A| + dn\log_b |A| \log n)$.

\noindent\underline{\textbf{Query Processing Cost.}} Query processing consists of query preparation at the QU and filter-and-refine search at the CS. At the QU side, range localization over $\mathcal{T}$ costs $O(\log_b |A|)$. Generating the DCPE and DCE trapdoors costs $O(d)$ and $O(d^2)$, respectively. Therefore, the query preparation cost at the QU is $O(d^2+\log_b |A|)$.

At the CS side, the filtering phase searches $|I_q|$ matched HNSW sub-indices. Since searching one HNSW sub-index costs $O(d\log n)$ in the worst case, the total sub-index search cost is $O(|I_q| d \log n)$. Each sub-index returns at most $k_0$ candidates, and the returned list is already ordered. Therefore, the merge cost of those lists is $O(k_0|I_q| \log |I_q|)$, and the CS then retains the top-$k'$ candidates for refinement. During refinement, the CS computes
DCE-based secure distances for these $k'$ candidates and maintains a max-heap of size $k$ to obtain the final top-$k$ results. Each DCE comparison costs $O(d)$, and each heap update costs $O(\log k)$. Therefore, the refinement phase costs $ O(k'(d + \log k))$. Overall, the CS side costs $O(|I_q|d\log n + k_0|I_q| \log |I_q| + k'(d+\log k))$.

During query processing, the QU and the CS exchange two messages. First, the QU sends the DCPE and DCE trapdoors, the matched node identifiers, and the parameter $k$ to the CS. The two trapdoors occupy $O(d)$ space, while the matched node identifiers occupy $O(|I_q|)$ space. Second, the CS returns the identifiers of the final top-$k$ results, which requires $O(k)$ space. Hence, the total communication overhead is $O(d+|I_q|+k)$.

\noindent\underline{\textbf{Space Complexity.}}
Let $N_T$ denote the number of nodes in the N-ary tree $\mathcal{T}$. At the QU side, the N-ary tree requires $O(N_T)$ space. The QU also stores $SK_{\mathsf{DCPE}}$ of space $O(1)$ and $SK_{\mathsf{DCE}}$ of space $O(d^2)$. 
Thus, the QU-side storage is $O(N_T+d^2)$. At the CS side, the DCPE and DCE ciphertexts require $O(nd)$ space. Since each object appears in the HNSW sub-indices along a root-to-leaf path of length $O(\log_b |A|)$, the total number of indexed object appearances is $O(n\log_b |A|)$. Given the HNSW degree parameter $M$, all HNSW sub-indices require $O(Mn\log_b |A|)$ space. Thus, the CS-side storage is $O(nd+Mn\log_b |A|)$.

\section{Security Analysis}
\label{sec:security}

In this section, we analyze the security of the proposed PP-RFANNS scheme against the honest-but-curious CS. Our scheme employs two existing encryption primitives, namely DCPE and DCE, to protect vectors and support encrypted distance comparisons. Instead of claiming a new simulation-based security result for the complete composition of DCPE, DCE, and the HNSW sub-indices, we rely on the security guarantees established for the
underlying primitives and explicitly characterize the additional information revealed by the complete PP-RFANNS protocol.


\subsection{Security of the Underlying Primitives}
\label{sec:primitive-security}

\noindent\underline{\textbf{DCPE Security.}}
The proposed scheme adopts the Scale-and-Perturb construction of DCPE~\cite{fuchsbauer2022approximate}. Given a vector $v$, its DCPE ciphertext is generated as $C_v^{\mathsf{DCPE}}=s\cdot v+\lambda_v$, where $s$ is a secret
scaling factor and $\lambda_v$ is a random perturbation vector. DCPE is designed to preserve approximate distance-comparison relationships while preventing direct recovery of plaintext vector coordinates under its original security model. The CS receives only the DCPE ciphertexts and does not obtain the secret parameters $s$ and $\beta$. Therefore, the plaintext vectors cannot be directly recovered without violating the security guarantee of DCPE. Nevertheless, DCPE intentionally reveals approximate geometric relationships. In particular, the CS can compute distances between DCPE ciphertexts and use them to construct and traverse the HNSW sub-indices. Hence, DCPE does not hide ciphertext-space distances or the approximate distance orders induced by these distances.

\noindent\underline{\textbf{DCE Security.}}
DCE~\cite{liu2025ppanns} enables the CS to compare the exact distances between two encrypted vectors and an encrypted query without revealing the underlying vectors or their exact plaintext distance values. The DCE construction is IND-KPA secure under its
specified distance-comparison leakage. Specifically, the CS learns the outcomes of the distance comparisons invoked during query processing, but does not learn the plaintext vector coordinates, query vector coordinates, or exact distance values.

In our scheme, DCPE and DCE use independent secret keys and randomness. A vector or query vector is encrypted separately under the two primitives. Consequently, information obtained from one ciphertext representation does not provide the secret key or randomness used by the other encryption scheme.

\subsection{Database Privacy}
\label{sec:database-privacy}

The CS receives the DCPE and DCE ciphertexts, the HNSW sub-indices, and the associated object identifiers. The CS does not receive the plaintext vectors in $V$ or the secret keys of either encryption scheme. Recovering the coordinates of a database vector from its DCPE or DCE ciphertext would therefore require violating the security guarantee of the corresponding primitive. The use of two encrypted representations does not by itself reveal the plaintext vector, since the representations are generated with independent keys and randomness.

The plaintext numerical attributes in $A$ are not outsourced to the CS. They are used by the DO to construct the range-aware N-ary tree, which is subsequently distributed to authorized QUs rather than the CS. Therefore, the CS does not directly observe the numerical value $a_i$ associated with an object.

However, each HNSW sub-index corresponds to an N-ary tree node, and the CS observes the object identifiers associated with each sub-index. Consequently, the CS learns the object-to-sub-index membership relation and can infer set-containment relationships among sub-indices, as well as the repeated occurrence of an object in multiple sub-indices along the tree hierarchy. This information may reveal that the attributes of certain objects fall within the same hidden interval, although it does not directly disclose the interval boundaries or the exact attribute values.

\subsection{Query Privacy}
\label{sec:query-privacy}

For a query $q=(v_q,[l_q,r_q],k)$, the QU locally generates the DCPE and DCE trapdoors
$T_q=(T_q^{\mathsf{DCPE}},T_q^{\mathsf{DCE}})$. Under the security guarantees of DCPE and DCE, the CS cannot directly recover the plaintext $v_q$ from the trapdoors. The DCPE
trapdoor nevertheless reveals the ciphertext-space distances and approximate ordering information required for HNSW search, while the DCE trapdoor reveals the outcomes of the exact comparisons performed during refinement.
The query range $[l_q,r_q]$ is processed locally by the QU over the N-ary tree. The numerical endpoints $l_q$ and $r_q$ are never sent to the CS. Instead, the QU sends the matched node set $I_q$, which identifies the HNSW sub-indices to be searched. Therefore, the CS does not directly learn the numerical query bounds.

However, because the CS knows the object identifiers contained in each matched sub-index, it can derive the set of objects $O_q=\bigcup_{tn\in I_q}O_{tn}$ in the query range. Thus, the proposed scheme hides the numerical semantics of the query range, but does not hide the access pattern of the objects covered by the range.

For repeated queries, the CS may additionally observe whether two queries access the same sub-indices or overlapping object sets. Consequently, equality, containment, and overlap relationships between different range queries may be inferred from their access
patterns, even though their numerical endpoints remain hidden.

\subsection{Leakage Characterization}
\label{sec:leakage-characterization}

We now summarize the information revealed to the CS during secure index construction and query processing. 

\noindent\underline{\textbf{Construction Leakage.}} After secure index construction, the CS observes
$
\mathcal{L}_{\mathsf{construct}}(O)
=
(n,d,\mathsf{pp},
\mathsf{Topo}(\mathcal{G}),
\mathsf{Mem}(\mathcal{G}),
\widetilde{\Delta}_{V})
$. Here, 
\ding{202} $n$ and $d$ denote the number and dimensionality of the database vectors, respectively; 
\ding{203} $\mathsf{pp}$ contains the public index and query parameters, including the HNSW parameters, the ciphertext dimensions, and the parameters $k_0$ and $k'$; \ding{204} $\mathsf{Topo}(\mathcal{G})$ denotes the number, sizes, levels, vertices, and edges of the outsourced HNSW sub-indices;
\ding{205} $\mathsf{Mem}(\mathcal{G}) =\{(tn, i)\mid o_i\in O_{tn}\}$ denotes the membership relationship between object identifiers and HNSW sub-indices; 
\ding{206} $\widetilde{\Delta}_{V}$ denotes the pairwise DCPE ciphertext-space distances among all encrypted database vectors, which are directly computable by the CS from the outsourced DCPE ciphertext set.


The components of $\mathcal{L}_{\mathsf{construct}}$ are jointly determined by the real index construction procedure. In particular, the HNSW topology depends on the sub-index memberships, DCPE distance relationships, insertion order, index parameters, and the
randomness used during HNSW construction. They should therefore be understood as a mutually consistent server-side view rather than independently generated leakage values.

\noindent\underline{\textbf{Query Leakage.}}
For a query $q=(v_q,[l_q,r_q],k)$, the CS observes
$
\mathcal{L}_{\mathsf{query}}(O,q)
=
(I_q,O_q,
\mathsf{AP}^{\mathsf{HNSW}}_q,
\widetilde{\Delta}_{V,q},
\mathsf{Cand}_q,
R_q,
\pi^{\mathsf{DCE}}_{R_q,q},
\mathsf{Out}_q)
$. Here, 
\ding{202} $I_q$ is the set of matched tree-node identifiers sent by the QU;
\ding{203} $O_q=\bigcup_{tn\in I_q}O_{tn}$ is the set of object identifiers covered by the matched sub-indices;
\ding{204} $\mathsf{AP}^{\mathsf{HNSW}}_q$ denotes the HNSW access pattern, including the sub-indices, vertices, and edges examined during filtering;
\ding{205} $\widetilde{\Delta}_{V,q}$ denotes the DCPE ciphertext-space distances between the query trapdoor and all encrypted database vectors, which are directly computable by the CS after receiving the DCPE query trapdoor.
\ding{206} $\mathsf{Cand}_q=\{R_{tn}\mid tn\in I_q\}$ denotes the local candidate lists returned from the matched HNSW sub-indices; 
\ding{207} $R_q$ is the candidate set retained after merging the local candidate lists and is a derived component of the filtering leakage;
\ding{208} $\pi^{\mathsf{DCE}}_{R_q,q}$ denotes the exact relative distance order revealed by the DCE comparisons over $R_q$;
\ding{209} $\mathsf{Out}_q$ denotes the identifiers of the final top-$k$ results.



The HNSW access pattern is determined by the matched sub-indices, their graph topology, the DCPE comparison outcomes, and the search parameters. It is listed explicitly because it is directly observable by the CS during query execution.

For a sequence of queries, repeated matched nodes, candidate identifiers, and result identifiers reveal equality and overlap patterns across query executions. These cross-query relations are derivable from the per-query leakage described above.




\subsection{Leakage Comparisons}

We compare our scheme with the secure adaptations instantiated in Section~3.4. All three adaptations use the same iSHE-based range-checking protocol. We focus on the view of the primary CS and assume that the auxiliary server does not collude with it.
Under their respective cryptographic assumptions, all four schemes protect plaintext vectors, numerical attributes, query vectors, query bounds, and exact distance values, subject to their explicitly characterized leakage. Secure Pre-filtering reveals the access pattern of its encrypted range index, the records subjected to iSHE-based checking, the range-checking outcomes, the complete qualified-object set, and the DCE comparison outcomes. Secure Post-filtering reveals the access pattern, the DCE comparison outcomes, and the candidate set of global PP-ANNS, and the iSHE checking outcomes over the retrieved candidates. PP-iRangeGraph reveals the DCPE-guided graph traversal trace, including the inspected vertices and edges, the iSHE range-checking outcomes for inspected neighbors, the coarse candidate set retained after traversal, and the DCE comparison outcomes over the final refinement set.

Our scheme avoids online encrypted range checking, but reveals the matched node set and, through object-to-sub-index memberships, the complete qualified-object set. Secure Pre-filtering reveals the same qualified-object set together with its encrypted-index access pattern, whereas Secure Post-filtering and PP-iRangeGraph reveal range-checking outcomes only for the candidates examined. Secure Post-filtering, PP-iRangeGraph, and our scheme all use DCPE for coarse filtering and DCE for exact refinement; their principal difference is the stage at which the range predicate restricts encrypted vector search. Thus, our scheme should not be viewed as strictly more secure than the baselines. Rather, it trades matched-node and qualified-set leakage for lower online query overhead and avoids the additional non-collusion assumption required by our two-server baseline instantiations.

\subsection{Summary}
\label{sec:security-summary}

Based on the security guarantees of DCPE and DCE, our scheme protects the plaintext coordinates of the vectors in $V$, the plaintext numerical values in $A$, the plaintext query vector $v_q$, the query bounds $l_q$ and $r_q$, the exact plaintext distance values, and the secret keys against the honest-but-curious CS.
At the same time, our scheme explicitly allows the leakage of the encrypted HNSW topology, object-to-sub-index memberships, DCPE ciphertext-space distance relationships, matched-node access patterns, HNSW traversal traces, candidate identifiers, DCE comparison outcomes, and final result identifiers. In particular, the CS can determine the qualified encrypted objects, although it does not directly learn the query range or the corresponding plaintext attribute values.

Therefore, the security guarantee of our scheme should be understood as follows: under the established security properties of DCPE and DCE, the scheme prevents the CS from directly recovering the protected plaintext database and query contents, subject to the explicitly characterized structural, access-pattern, and distance-comparison leakage.
Our analysis assumes a static database and an honest-but-curious CS. Dynamic updates, malicious-server behavior, and result verifiability are outside the scope of this work.

\section{Experiments}
\label{sec:exp}

In this section, we conduct extensive experiments to demonstrate the superiority of our scheme over the baselines. We first present the experimental settings and then report our findings.

\begin{figure*}[t]
  \centering
  \includegraphics[width=0.6\textwidth]{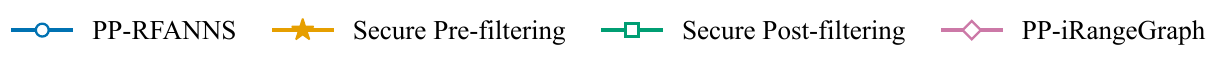}
  \vspace{0.3em}
  \setlength{\tabcolsep}{2pt}
  \renewcommand{\arraystretch}{0.9}
  \resizebox{\textwidth}{!}{%
  \begin{tabular}{@{}ccccc@{}}
    \includegraphics[width=0.16\textwidth]{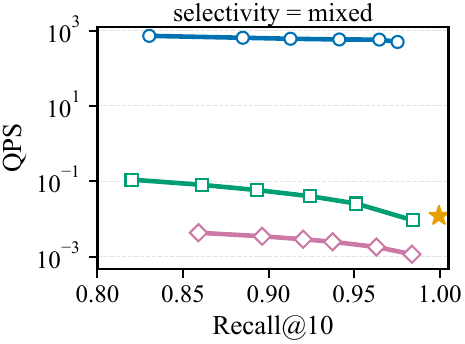} &
    \includegraphics[width=0.16\textwidth]{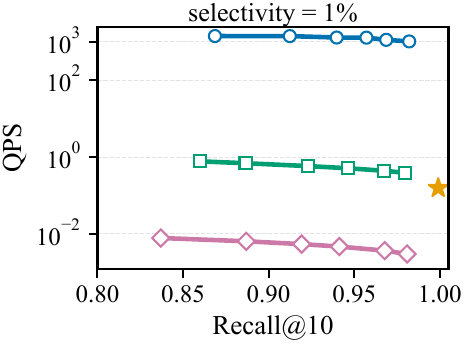} &
    \includegraphics[width=0.16\textwidth]{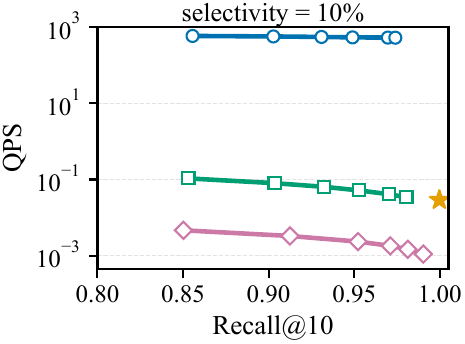} &
    \includegraphics[width=0.16\textwidth]{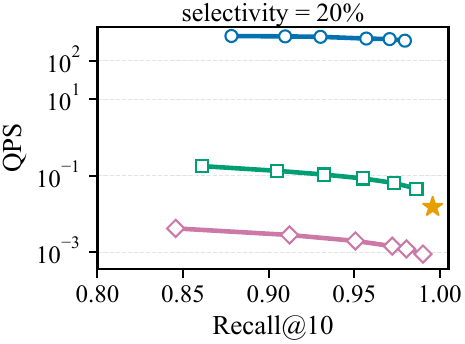} &
    \includegraphics[width=0.16\textwidth]{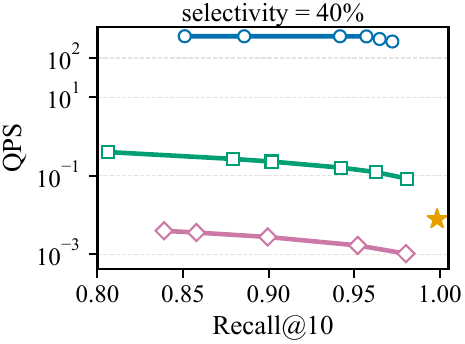} \\[-0.2em]
    \multicolumn{5}{c}{\footnotesize \textbf{(a) Sift1M}} \\[0.4em]

    \includegraphics[width=0.16\textwidth]{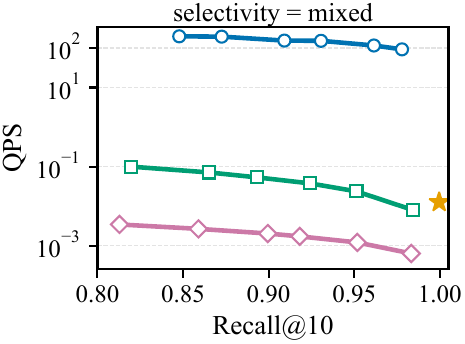} &
    \includegraphics[width=0.16\textwidth]{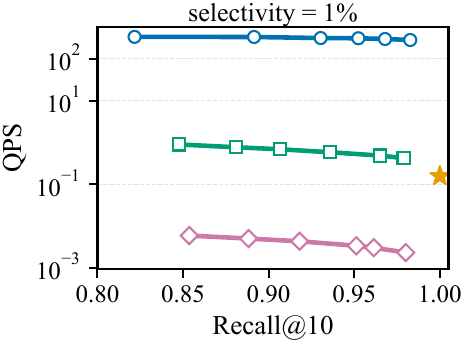} &
    \includegraphics[width=0.16\textwidth]{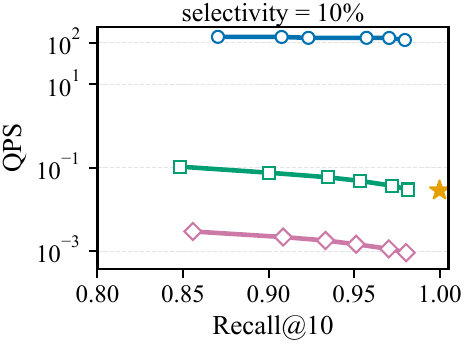} &
    \includegraphics[width=0.16\textwidth]{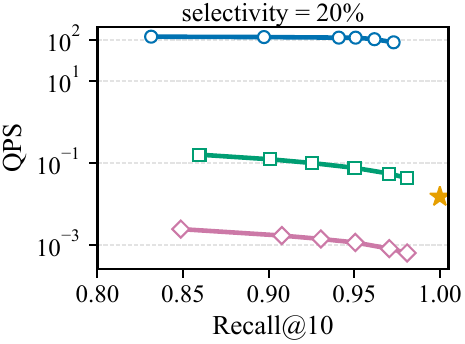} &
    \includegraphics[width=0.16\textwidth]{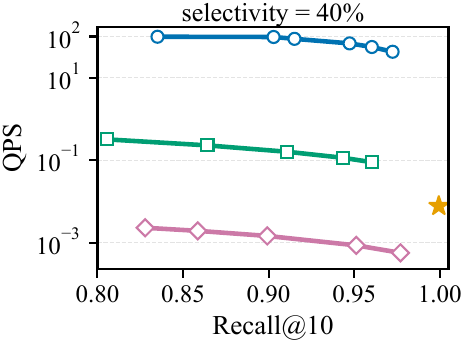} \\[-0.2em]
    \multicolumn{5}{c}{\footnotesize \textbf{(b) Gist}} \\[0.4em]

    \includegraphics[width=0.16\textwidth]{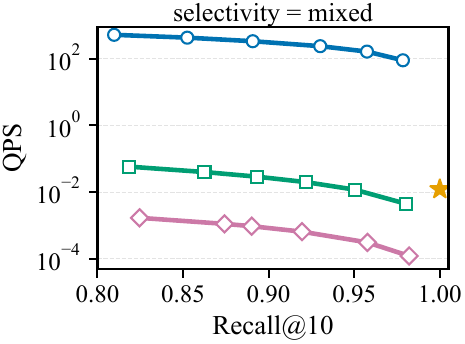} &
    \includegraphics[width=0.16\textwidth]{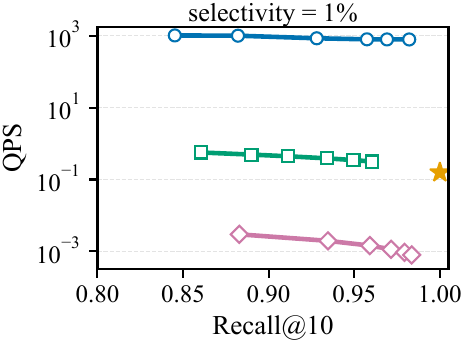} &
    \includegraphics[width=0.16\textwidth]{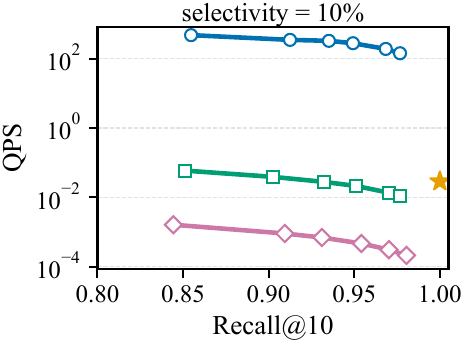} &
    \includegraphics[width=0.16\textwidth]{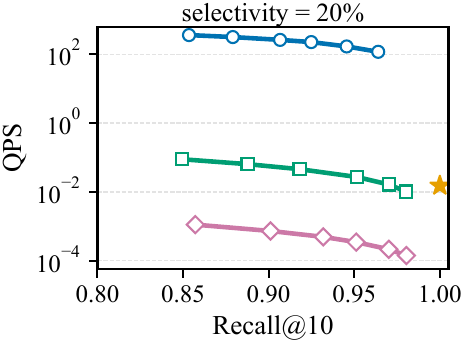} &
    \includegraphics[width=0.16\textwidth]{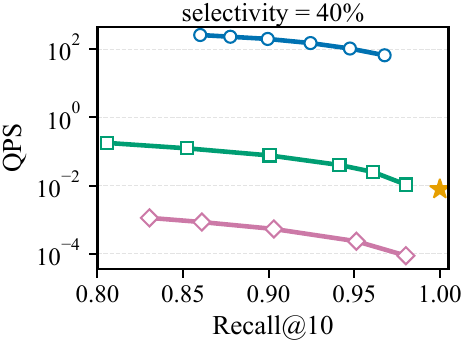} \\[-0.2em]
    \multicolumn{5}{c}{\footnotesize \textbf{(c) GloVe}} \\[0.4em]

    \includegraphics[width=0.16\textwidth]{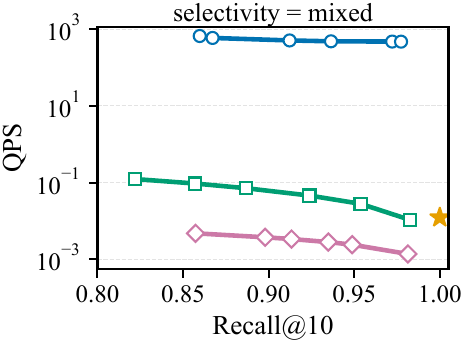} &
    \includegraphics[width=0.16\textwidth]{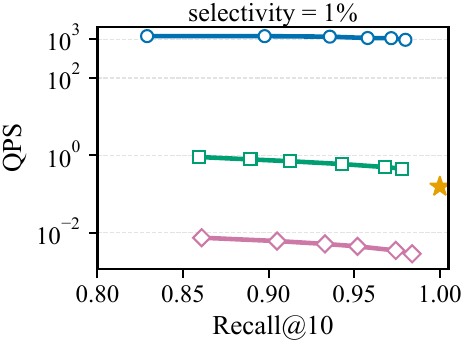} &
    \includegraphics[width=0.16\textwidth]{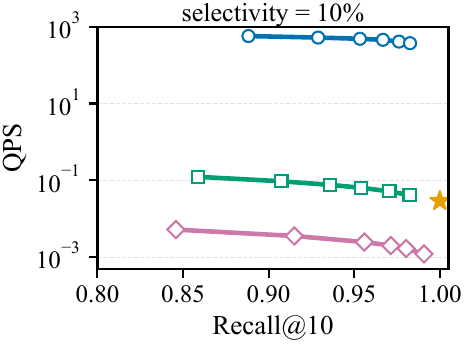} &
    \includegraphics[width=0.16\textwidth]{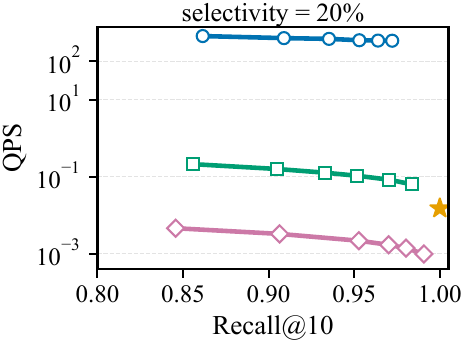} &
    \includegraphics[width=0.16\textwidth]{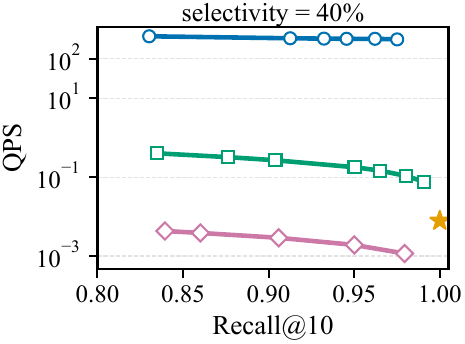} \\[-0.2em]
    \multicolumn{5}{c}{\footnotesize \textbf{(d) Deep1M}}
  \end{tabular}%
  }
  \vspace{-3ex}
  \caption{Comparisons of search performance of PP-RFANNS queries. (Exp. 1)}
  \label{fig:baseline_fraction}
\end{figure*}

\subsection{Experimental Settings}
\label{ssec:exp_set}


\noindent\underline{\textbf{Datasets.}}
We construct object sets by associating each vector with a numerical attribute. Specifically, we use four widely adopted $k$-ANNS benchmark datasets: Sift1M~\footnote{\label{fn:texmex}\url{http://corpus-texmex.irisa.fr/}}, Gist~\footref{fn:texmex}, GloVe~\footnote{\url{https://nlp.stanford.edu/projects/glove/}}, and Deep1M~\footnote{\url{http://sites.skoltech.ru/compvision/noimi/}}. Each vector is assigned an integer attribute independently and uniformly sampled from the domain $[1,10000]$. For each query vector, we randomly generate an interval within the same domain as the range predicate. Table~\ref{tab:datasets} summarizes the dataset statistics.

For the scalability evaluation, we further construct datasets containing different numbers of vectors randomly sampled from Sift1B~\footref{fn:texmex} and Deep1B~\footnote{\url{https://disk.yandex.ru/d/11eDCm7Dsn9GA/}}. Their numerical attributes and query ranges are generated following the same procedure. Unless otherwise specified, we use mixed-selectivity setting by default, where each query range is generated by first uniformly sampling its left endpoint from the attribute domain and then uniformly sampling its right endpoint from the remaining domain to the right.

\begin{table}[t]
  \caption{Dataset statistics.}
  \vspace{-2ex}
  \label{tab:datasets}
  \centering
  \begin{tabular}{|l|c|c|c|c|}
    \hline
    Dataset & \#dimensions & \#objects & \#queries \\
    \hline
    Sift1M & 128 & 1,000,000 & 10,000 \\
    \hline
    Gist & 960 & 1,000,000 & 1,000 \\
    \hline
    GloVe & 100 & 1,183,514 & 10,000 \\
    \hline
    Deep1M & 96 & 1,000,000 & 10,000 \\
    \hline
  \end{tabular}
\end{table}

\noindent\underline{\textbf{Performance Metrics.}}
We evaluate query efficiency using queries per second (QPS) and the query latency, and estimate the accuracy by $Recall@k$. QPS is calculated as the total number of queries divided by their total processing time. For a query $q$, let $N^*(q)$ denote the exact $k$-nearest neighbors among the vectors satisfying its range predicate, and let $N(q)$ denote the results returned by a PP-RFANNS method. The recall of $q$ is defined as $\mathrm{Recall}@k(q) = \frac{\left|N^*(q)\cap N(q)\right|}{k}$. We report $Recall@k$ averaged over all queries. Unless otherwise specified, we set $k=10$.


\noindent\underline{\textbf{Baselines.}}
We compare our scheme with the secure adaptations introduced in Section~\ref{sec:limit}: Secure Pre-filtering, Secure Post-filtering, and PP-iRangeGraph. All methods operate in an outsourced encrypted-search setting. 
For a fair comparison, all methods use the same datasets, queries, and settings for the cryptographic primitives shared among them. Following prior work~\cite{liu2025ppanns}, we set $M=40$ and $\mathit{efConstruction}=600$ for all HNSW-based indices. We vary the principal search parameters of each approximate method to obtain its QPS–$Recall@k$ trade-offs. Specifically, PP-RFANNS varies $k_0$, $k’$, and the HNSW search breadth; Secure Post-filtering varies the search breadth of its global encrypted HNSW index and the number $k’$ of retrieved candidates; and PP-iRangeGraph varies its sub-index search breadth. Secure Pre-filtering exhaustively evaluates all qualified vectors using DCE and therefore returns the exact top-$k$ results.

\noindent\underline{\textbf{Computing Environments.}}
All experiments are conducted on a server equipped with two 10-core Intel(R) Xeon(R) Silver 4210 CPUs and 256~GiB of RAM, running CentOS 7.9.2009. Three baselines additionally use a non-colluding auxiliary server with the same hardware configuration to support iSHE-based secure range checking, whereas PP-RFANNS operates with a single cloud server. All methods are implemented in \texttt{C++}, and the query-processing operations on each server are executed using a single thread unless otherwise specified. Under this setting, one end-to-end iSHE-based range check takes $0.0312$ seconds on average, including all cryptographic operations and inter-server communication.

\subsection{Experimental Results}
\label{ssec:exp_res}

\noindent\underline{\textbf{Exp. 1: Search Performance of PP-RFANNS Queries.}}
We compare PP-RFANNS with the three baselines under mixed and fixed query selectivities from 1\% to 40\%. As shown in Figure~\ref{fig:baseline_fraction}, PP-RFANNS consistently achieves the best QPS–Recall trade-off. Even at 40\% selectivity, where Secure Post-filtering is relatively favorable, PP-RFANNS still achieves 595x to 2600x speedups at $Recall@10 = 0.95$. 
Secure Pre-filtering achieves $Recall@10 = 1$, but its QPS is always less than 0.2. Although its secure range index reduces the number of iSHE checks, the remaining checks still incur substantial computation and interaction overhead. It also performs exhaustive DCE-based search over all qualified vectors, making it increasingly expensive at high selectivity. PP-iRangeGraph repeatedly invokes iSHE during graph traversal, so its QPS is generally below 0.01. In contrast, PP-RFANNS reaches 68--1,288 QPS at $Recall@10=0.95$.


All three baselines incur online iSHE range-checking and inter-server interaction. PP-RFANNS instead localizes the range at the QU and searches only the matched encrypted sub-indices on a single CS, achieving better QPS--Recall trade-offs across different selectivities.



\begin{figure}[!t]
  \centering
  \includegraphics[width=0.36\textwidth]{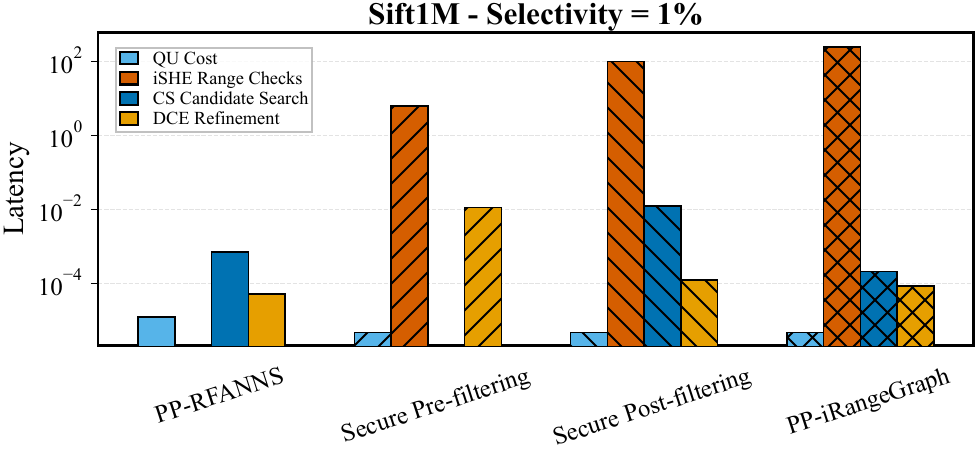}
  \hfill
  \includegraphics[width=0.36\textwidth]{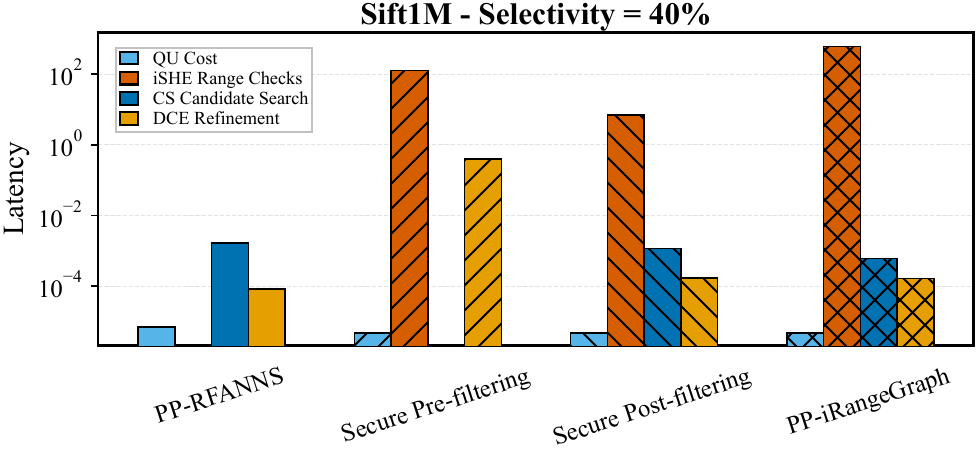}
  \vspace{-2ex}
  \caption{Cost decomposition on Sift1M at $Recall@10=0.95$ (Exp. 2)}
  \label{fig:time_breakdown}
\end{figure}




\noindent\textbf{Exp. 2: Cost Decomposition of PP-RFANNS Queries.} We decompose the per-query latency of each method to identify its dominant cost components and explain the observed performance differences.
Figure~\ref{fig:time_breakdown} reports the latency breakdown on Sift1M at $Recall@10=0.95$. The decomposition shows that the query latency of all three baselines is
dominated by iSHE-based secure range checking. Secure Pre-filtering uses iSHE to identify the qualified objects before vector search, Secure Post-filtering applies iSHE to the candidates retrieved from the global encrypted HNSW index, and PP-iRangeGraph repeatedly invokes iSHE for newly inspected neighbors during graph traversal. These online checks
incur expensive homomorphic computation, decryption, and inter-server interaction. In contrast, PP-RFANNS performs range localization locally at the QU before server-side vector search and therefore eliminates this dominant online iSHE cost. This result directly demonstrates the advantage of decoupling range localization from encrypted vector search.

Among the remaining CS-side components, the vector search costs differ across the baselines. Secure Pre-filtering directly performs exhaustive DCE-based distance comparisons over all qualified vectors. Its DCE cost therefore grows with the number of objects satisfying the range predicate and becomes substantial at high selectivity. In contrast, others use DCPE for coarse candidate search and DCE only for final refinement. Their DCPE-based search repeatedly evaluates ciphertext-space distances during HNSW or graph traversal and may inspect many vertices before obtaining a sufficiently accurate candidate set. Although an individual DCE comparison is more expensive, DCE is applied only to the much
smaller refinement set retained after coarse retrieval, so its total cost remains relatively low.

\begin{table}[t]
\centering
\caption{Index cost: each entry reports index building time in seconds / index size in GiB (Exp. 3)}
\label{tb:index_cost}
\small
\setlength{\tabcolsep}{3.5pt}
\renewcommand{\arraystretch}{1.05}
\begin{tabular}{lcccc}
\toprule
\textbf{Method}
& \textbf{Sift1M}
& \textbf{Gist}
& \textbf{GloVe}
& \textbf{Deep1M} \\
\midrule
PP-RFANNS
& 418 / 2.51
& 2531 / 2.51
& 768 / 5.39
& 411 / 2.51 \\

Secure Pre
& 0.03 / 0.005
& 0.03 / 0.005
& 0.03 / 0.006
& 0.03 / 0.005 \\

Secure Post
& 150 / 0.79
& 1139 / 3.89
& 204 / 0.82
& 138 / 0.68 \\

PP-iRangeGraph
& 3071 / 1.49
& 10373 / 1.40
& 4780 / 1.47
& 2865 / 1.62 \\

DCPE Encryption
& 6 / 0.47
& 31 / 3.57
& 6 / 0.44
& 5 / 0.35 \\

DCE Encryption
& 29 / 8.10
& 287 / 57.69
& 26 / 7.61
& 18 / 6.19 \\

\bottomrule
\end{tabular}
\end{table}

\noindent\underline{\textbf{Exp. 3: Index Cost.}} We evaluate the index cost of each method in terms of \emph{building time} and \emph{index size}. Table~\ref{tb:index_cost} reports both metrics on the four datasets. Since DCPE and DCE encryption are independent of index construction, their encryption time and ciphertext size are reported separately in the last two rows.

Secure Pre-filtering and Secure Post-filtering incur relatively low index costs because the former only builds a lightweight B-tree, whereas the latter maintains a single global HNSW index. In contrast, PP-RFANNS and PP-iRangeGraph construct hierarchical range-aware graph indices and therefore require substantially more construction time and storage. 
Between the two hierarchical methods, PP-RFANNS is faster to build than PP-iRangeGraph. This improvement mainly results from its shallower N-ary-tree hierarchy. The depths of PP-RFANNS are 7, 7, 14, and 7 on Sift1M, Gist, GloVe, and Deep1M, respectively, compared with 20, 20, 21, and 20 for PP-iRangeGraph. Because graph indices are constructed repeatedly along root-to-leaf paths, the reduced hierarchy depth considerably lowers the overall construction workload.  PP-RFANNS, however, requires more index space than PP-iRangeGraph, due to their underlying graph-index implementations. Overall, PP-RFANNS trades additional
storage for substantially faster index construction.

\begin{figure*}[t]
    \centering
    \vspace{-1em}
    \makebox[\textwidth][c]{%
        \includegraphics[width=0.22\textwidth]{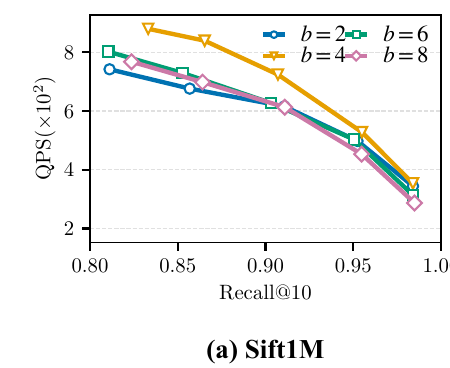}\hspace{0.002\textwidth}%
        \includegraphics[width=0.22\textwidth]{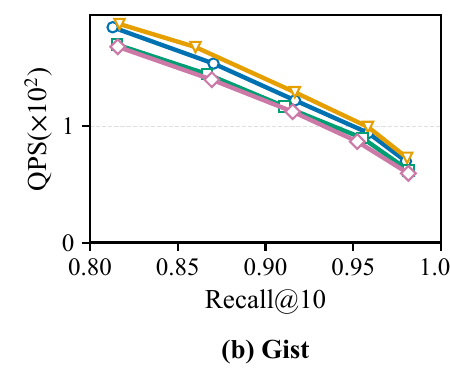}\hspace{0.002\textwidth}%
        \includegraphics[width=0.22\textwidth]{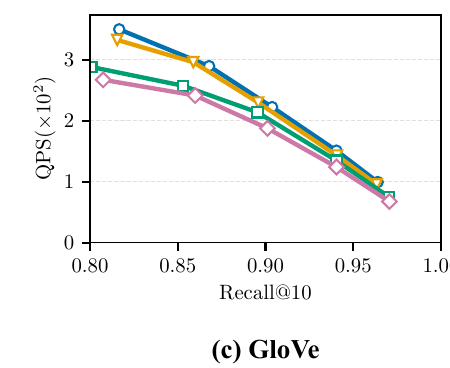}\hspace{0.002\textwidth}%
        \includegraphics[width=0.22\textwidth]{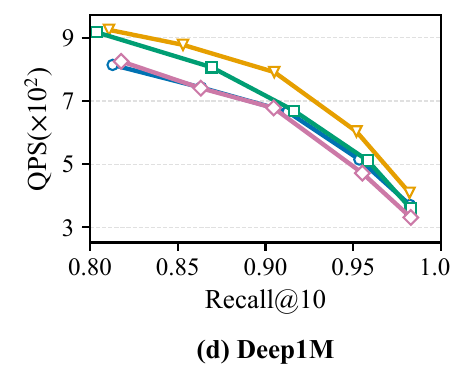}%
    }
    \vspace{-5ex}
    \caption{Effects of the branching factor $b$ on query performance (Exp. 4)}
    \label{fig:effect_N}
\end{figure*}


\noindent\textbf{Exp. 4: Effect of the Branching Factor $b$.}
The branching factor $b$ controls the structure of the N-ary tree and the decomposition of a query range. Increasing $b$ reduces the tree height, but may produce more matched nodes and consequently require the CS to search more HNSW sub-indices. Figure~\ref{fig:effect_N} evaluates the effect of $b$ on the QPS--Recall trade-off. On Sift1M, Gist, and Deep1M, query performance first improves and then degrades as $b$ increases, with $b=4$ generally achieving the best trade-off. On GloVe, performance decreases with increasing $b$, indicating that the additional sub-index searches outweigh the benefit of a shallower tree. Based on these results, we set $b$ to $4$, $4$, $2$, and $4$ for Sift1M, Gist, GloVe, and Deep1M, respectively, in Exp. 1.

\begin{figure*}[t]
    \centering
    \vspace{-1em}
    \makebox[\textwidth][c]{%
        \includegraphics[width=0.22\textwidth]{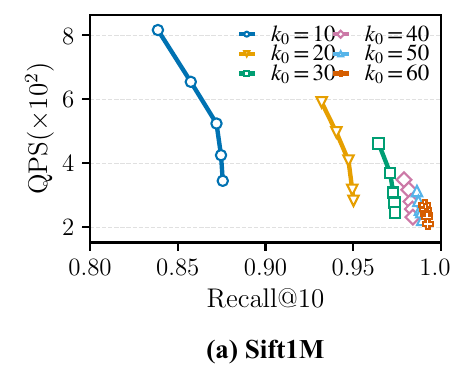}\hspace{0.002\textwidth}%
        \includegraphics[width=0.22\textwidth]{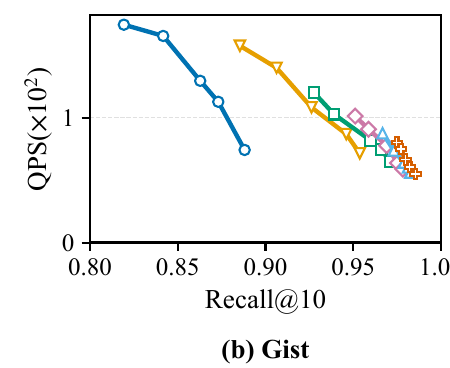}\hspace{0.002\textwidth}%
        \includegraphics[width=0.22\textwidth]{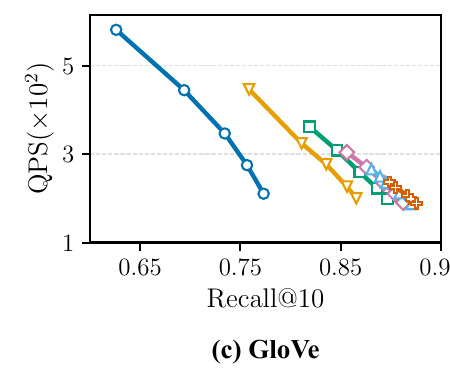}\hspace{0.002\textwidth}%
        \includegraphics[width=0.22\textwidth]{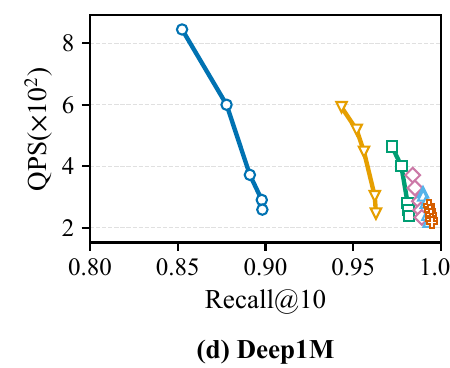}%
    }
    \vspace{-5ex}
    \caption{Effects of the candidate number $k_0$ on query performance (Exp. 5)}
    \label{fig:effect_k0}
\end{figure*}


\noindent\underline{\textbf{Exp. 5: Effect of the Search Parameter $k_0$.}}
The parameter $k_0$ controls the number of candidates retrieved from each matched HNSW sub-index during the filtering phase. As shown in Figure~\ref{fig:effect_k0}, increasing $k_0$ generally improves $Recall@10$ but reduces QPS. A larger $k_0$ expands the candidate pool and increases the likelihood of retaining the true nearest neighbors, at the cost of additional sub-index search, candidate merging, and refinement overhead. These results demonstrate the expected QPS--Recall trade-off controlled by $k_0$.

\begin{figure*}[t]
    \centering
    \vspace{-1em}
    \makebox[\textwidth][c]{%
        \includegraphics[width=0.22\textwidth]{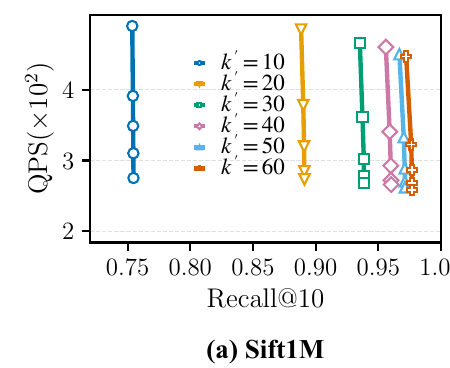}\hspace{0.002\textwidth}%
        \includegraphics[width=0.22\textwidth]{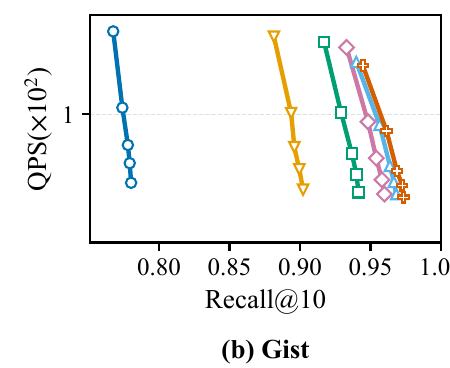}\hspace{0.002\textwidth}%
        \includegraphics[width=0.22\textwidth]{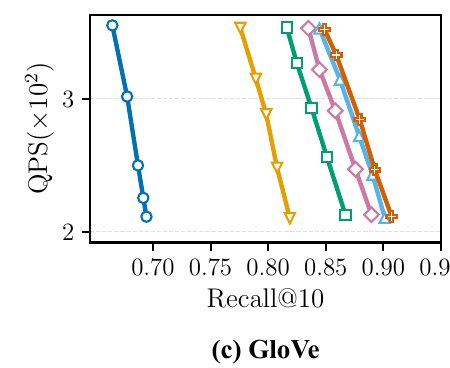}\hspace{0.002\textwidth}%
        \includegraphics[width=0.22\textwidth]{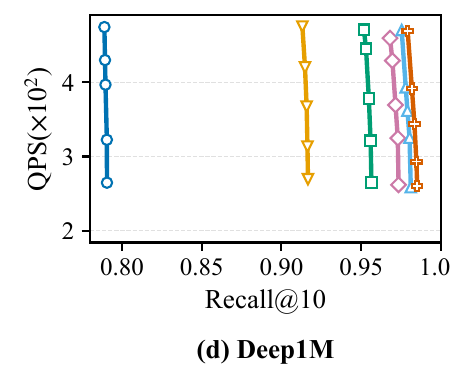}%
    }
    \vspace{-5ex}
    \caption{Effects of the parameter $k'$ on query performance (Exp. 6)}
    \label{fig:effect_k_prime}
\end{figure*}


\noindent\underline{\textbf{Exp. 6: Effect of the Search Parameter $k^\prime$.}}
The parameter $k^\prime$ controls the number of coarse candidates retained after merging the local candidate lists for DCE-based refinement. A larger $k^\prime$ increases the likelihood that the true nearest neighbors are included in the refinement set, but also enlarges the number of expensive DCE comparisons. To evaluate this trade-off, we fix $k_0=40$ and vary $k^\prime$. As shown in Figure~\ref{fig:effect_k_prime}, $Recall@10$ generally improves as $k^\prime$ increases and gradually converges, while the additional refinement cost reduces query efficiency. This result indicates that most true nearest neighbors appear near the top of the merged candidate list. GloVe exhibits a larger recall variation than the other datasets, suggesting that its candidate quality is more sensitive to $k^\prime$.

\begin{figure}[t]
  \centering
  \includegraphics[width=0.46\linewidth]{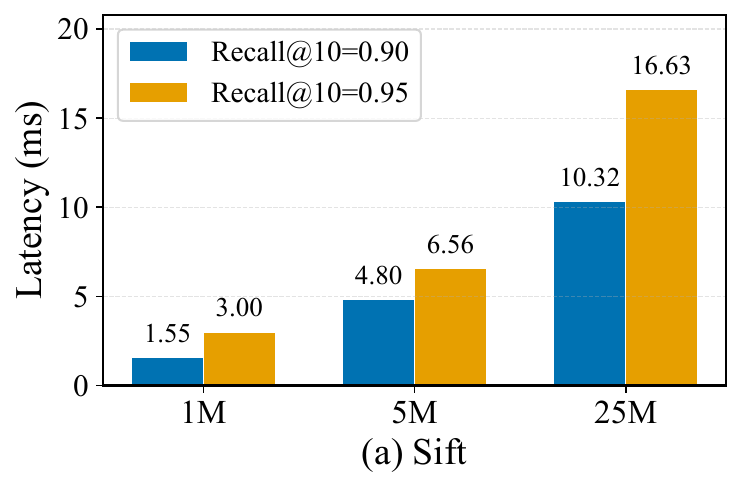}
  \includegraphics[width=0.46\linewidth]{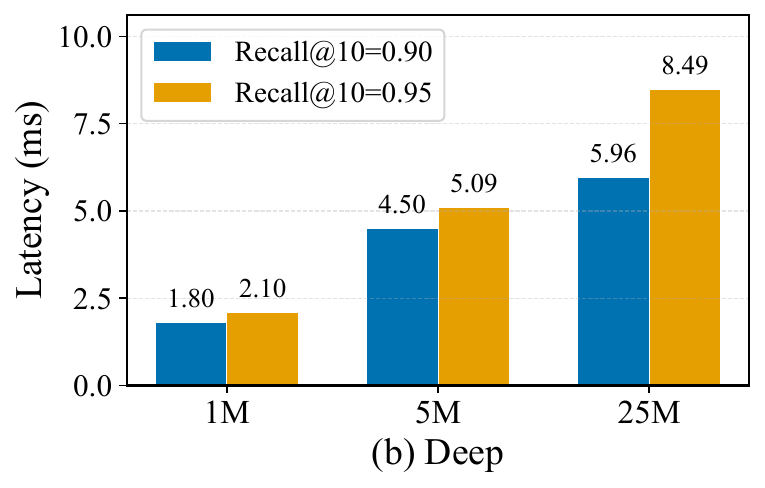}

  \vspace{-0.2em}

  \includegraphics[width=0.46\linewidth]{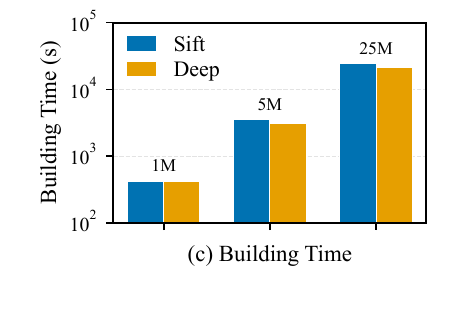}
  \includegraphics[width=0.46\linewidth]{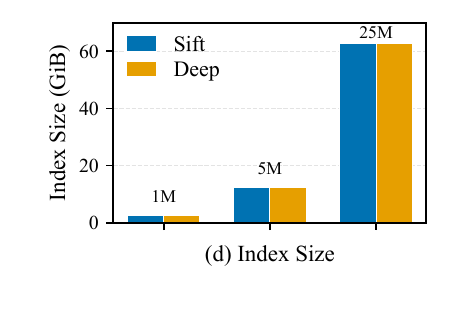}

  \vspace{-2ex}
  \caption{The scalability of our scheme (Exp. 7).}
  \label{fig:scalability}
  \vspace{-2ex}
\end{figure}

\noindent\underline{\textbf{Exp. 7: Scalability Test.}}
We evaluate the scalability of PP-RFANNS on random samples of Sift1B and Deep1B containing 1M, 5M, and 25M vectors. The attribute values and range-filtered queries are generated following the procedure described in the experimental setup. Figure~\ref{fig:scalability} reports both the online search performance and the offline index cost. Figures~\ref{fig:scalability}(a) and~\ref{fig:scalability}(b) show the query latency at Recall@$10$ values of 0.90 and 0.95. On both datasets, the latency increases gradually with the database size and remains below 17 ms even at 25M vectors. Figures~\ref{fig:scalability}(c) and~\ref{fig:scalability}(d) report the index construction time and index size, respectively. The index size increases from about 2.5 GiB at 1M vectors to about 62.6 GiB at 25M vectors. This growth is expected because the N-ary tree has $h=O(\log_b|A|)$ levels, and each object is indexed by the HNSW sub-indices along its root-to-leaf path, resulting in $O(n\log_b|A|)$ total indexed object occurrences. The construction time also grows faster than linearly because each object must be inserted into HNSW sub-indices at multiple tree levels. Nevertheless, index construction is performed only once offline. Overall, PP-RFANNS maintains low query latency as the database grows, while exhibiting the expected hierarchical-index construction and storage costs.

\vspace{-2ex}
\section{Conclusion}
In this paper, we study privacy-preserving range-filtered approximate nearest neighbor search over outsourced encrypted data. We proposed a PP-RFANNS scheme that combines an N-ary tree--HNSW hybrid index with DCPE-based filtering and DCE-based refinement. The scheme performs range localization at the query user and restricts encrypted vector search to the matched HNSW sub-indices, thereby avoiding online encrypted range checking.
We analyze its cost and security under the honest-but-curious model. Experiments show that PP-RFANNS achieves better QPS--Recall trade-offs than secure pre-filtering, secure post-filtering, and PP-iRangeGraph across different query selectivities, and scales effectively to larger datasets.


\bibliographystyle{ACM-Reference-Format}
\bibliography{main}

\end{document}